\documentclass[12pt]{article}
\pdfoutput=1
\usepackage{amsmath}
\usepackage{amsfonts}
\usepackage{mitpress}
\usepackage{comment}
\usepackage[toc,page, title]{appendix}
\usepackage{graphicx}
\usepackage{color}
\usepackage{enumerate}
\usepackage{pdfpages}

\newcommand{\p}{\partial}
\newcommand{\spt}{\mathop{\rm spt}}

\newtheorem{theorem}{Theorem}

\newtheorem{corollary}[theorem]{Corollary}

\newtheorem{definition}[theorem]{Definition}
\newtheorem{example}[theorem]{Example}

\newtheorem{lemma}[theorem]{Lemma}

\newtheorem{proposition}[theorem]{Proposition}
\newtheorem{remark}[theorem]{Remark}

\newenvironment{proof}[1][Proof]{\noindent\textbf{#1.} }{\ \rule{0.5em}{0.5em}}
\newdimen\dummy
\dummy=\oddsidemargin
\begin{document}

\title{Multidimensional matching\thanks{%
The authors are grateful to Toronto's Fields' Institute for the Mathematical
Sciences for its kind hospitality during part of this work. They acknowledge
partial support of RJM's research by Natural Sciences and Engineering
Research Council of Canada Grant 217006-08 and -15. Chiappori gratefully
acknowledges financial support from the NSF (Award 1124277). Pass is pleased
to acknowledge support from Natural Sciences and Engineering Research
Council of Canada Grant 412779-2012 and a University of Alberta start-up
grant. \hskip0.75in \copyright \today}}
\author{Pierre-Andr\'{e} Chiappori\thanks{%
Department of Economics, Columbia University, New York, USA
pc2167@columbia.edu}, Robert McCann\thanks{%
Department of Mathematics, University of Toronto, Toronto, Ontario, Canada
mccann@math.toronto.edu} and Brendan Pass\thanks{%
Department of Mathematical and Statistical Sciences, University of Alberta,
Edmonton, Alberta, Canada pass@ualberta.ca.}}
\maketitle

\begin{abstract}
We present a general analysis of multidimensional matching problems with
transferable utility, paying particular attention to the case in which the
dimensions of heterogeneity on the two sides of the market are unequal. A
particular emphasis is put on problems where agents on one side of the
market are multidimensional and agents on the other side are
uni-dimensional; we describe a general approach to solve such problems.
Lastly, we analyze several examples, including an hedonic model with
differentiated products, a marriage market model where wives are
differentiated in income and fertility, and a competitive variation of the
Rochet-Chon\'{e} problem. In the latter example, we show that the bunching
phenomena, observed by Rochet and Chon\'{e} in the monopoly context, do not
occur in the competitive context.
\end{abstract}

\section{Introduction}

Matching problems under transferable utility have attracted considerable
attention in recent years within economic theory. The general goal is to
understand the stable equilibria of matches between distributions of agents
on two sides of a `market' (for example, husbands and wives in the marriage
market, CEOs and firms in the labor market, producers and consumers in a
market for differentiated commodities, etc.), as well as the resulting
division of surplus between partners. Until recently, most work has focused
on the setting in which a single characteristic is used to distinguish
between the agents on each side; for example, in the marriage market,
several models assume that individuals differ only by their income (or their
human capital). These models have the advantage of being analytically
tractable, and often allow explicit closed form solutions. Under the
classical Spence-Mirrlees condition, the only stable matching is the
positive assortative one, the nature of which (`who marries whom') is
directly determined by the underlying distributions of the male and female
characteristic. However, these one dimensional models are unsatisfactory in
many situations, as both casual empiricism and factual evidence indicate
that agents often match on several traits. In the marriage market, for
instance, the suitability of a potential marriage between a woman and man
typically depends on several characteristics of both, including income and
education, but also age, tastes, ethnic background, physical attractiveness,
etc.

It is therefore important to study and understand \emph{multidimensional}
matching problems, in which agents on both side of the market are
differentiated using several characteristics. These models have garnered
increasing visibility in recent years, due to their wider applicability and
flexibility, but their introduction brings forth serious theoretical
challenges. The nature of the equilibrium matching is more interesting but
also more complex; in contrast to the one dimensional case, it is no longer
determined by the sole knowledge of the distributions of individual
characteristics, even under (a generalization of) the Spence-Mirrlees
condition. From a more technical perspective, it is generally not possible
to derive closed form solutions; and discretising matching problems leads to
a linear program, which often become numerically unwieldy when type spaces
are multidimensional.

The purpose of this paper is to provide a general characterization of
multidimensional matching models, in terms of existence, uniqueness and
qualitative properties of stable matches. Since the work of \cite{SS} in the
discrete setting, and \cite{GOZ} in the continuum, it has been understood
that transferable utility matching is equivalent to a variational problem;
this problem is known in the mathematics literature as the \textit{%
Monge-Kantorovich optimal transport} problem. Our first goal is to show that
the considerable existing literature on multidimensional optimal transport%
\footnote{%
For a general survey, see for instance \cite{V2} and \cite{Santambrogio15p}.}
can be exploited to derive conditions under which stable matches are unique
and pure, and to understand their properties.

We put a particular emphasis on the case in which the dimensions of
heterogeneity on the two sides of the market are unequal (say, $m>n$). These
sorts of problems have received relatively little attention from the
mathematics community, but are quite natural economically; the dimension
essentially reflects the number of attributes used to distinguish between
agents and there is no compelling reason in general to expect this number to
coincide for agents on the two different sides of the market (say, for
consumers and producers). A typical pattern emerges in these situations,
since for one side of the market (the one with a lower dimension), identical
agents are typically matched with a continuum of different partners. We
explore the properties of the `indifference sets' thus defined, and argue
that since such indifference sets can often be empirically recovered, these
properties can provide testable consequences of multidimensional matching
theory. 

Of specific interest are the so-called `multi-to-one dimensional matching
problems', in which agents on one side of the market are assumed to be
multidimensional, while those on the other side are unidimensional. They
include an economically important class of examples (see, for instance, \cite%
{COQ} and \cite{L}), for which one can obtain explicitly the stable
matchings. In this context, we describe a general approach aimed at
characterizing the equilibrium matching. We provide a robust methodology
that allows, under suitable conditions, to explicitly characterize its
solutions. We discuss some interesting features that the indifference sets
exhibit, and which are typically absent in purely unidimensional problems.
For instance, the optimal mapping may be discontinuous, and so women of
similar types may marry men of very different types.

Lastly, we consider three specific examples that illustrate different
potential applications of matching models. One is a standard, hedonic model
of the type used, in particular, in the empirical IO literature. While these
models typically assume imperfect competition, we show how a competitive
version can be analyzed. Under standard conditions, existence, but also
uniqueness and purity, obtain naturally. In particular, we show that the
mathematical notion of purity has in this context a natural interpretation
in terms of `bunching'. In addition, it is in general possible, using
matching or optimal transportations techniques, to derive closed form
solutions for the pricing schedule; we illustrate how this can be done using
specific distributions.

Our second example considers a context where women are characterized by two
traits (socio-economic status, from now on SES, and fertility) whereas men
only differ by their SES - a case recently studied by \cite{L}. Then it may
be the case that small changes in a middle class woman's SES (while her
fertility remains constant) result in a large change in her husband's
income. Again, we provide a general characterization of the solution and
describe some of its qualitative features.

Lastly, the multi-to-one dimensional framework naturally leads to
investigating the relationship between matching models and principal-agents
problems under multidimensional asymmetric information. While this general
question remains widely open (and quite challenging), we illustrate our
contribution by considering a competitive variant of the seminal model of 
\cite{rc}, in which goods can be produced by a competitive and $1$%
-dimensionally heterogeneous set of producers, rather than a single
monopolist. Unlike the original Rochet-Chon\'{e} framework, this model is
equivalent to a matching problem of our form, and can be solved explicitly
by the methods described above. We provide a full characterization of the
resulting equilibrium price schedule. In particular, we show that in our
competitive framework, and in contrast to the original monopolist setting,
there is never bunching: consumers of different types \textit{always} buy
goods with different characteristics. In other words, the strange bunching
patterns emphasized by Rochet and Chon\'{e} do not appear to be
intrinsically linked to the multidimensional nature of the adverse selection
problem; rather, they are due to the distorsions created by the presence of
a monopolist producer.

\section{Multidimensional matching under transferable utility: basic
properties}

\subsection{General framework and basic results}

\subsubsection{The model}

We consider sets $X\subseteq {\mathbf{R}}^{m}$ and $Y\subseteq {\mathbf{R}}%
^{n}$, parametrizing populations of agents on two sides of a market. In what
follows, we shall stick to the marriage market interpretation (so that $X$
and $Y$ will denote the set of potential wives and husbands respectively),
although alternative interpretations are obviously possible. They are
distributed according to probability measures $\mu $ on $X$ and $\nu $ on $Y$%
, respectively. In the transferable utility framework, a potential matching
of agents $x\in X$ and $y\in Y$ generates a combined surplus $s(x,y)$, where 
$s:X\times Y\rightarrow {\mathbf{R}}$. This surplus can be divided in any
way between the agents $x$ and $y$. For simplicity, we assume that $s$ and
its derivatives are smooth and bounded unless otherwise remarked; many of
the results we describe can also be extended to surpluses with less
smoothness, as in \cite{cmn} and \cite{NoldekeSamuelson15p} for example.

A \textit{matching} is characterized by a probability measure $\gamma $ on
the product $X\times Y$, whose marginals are $\mu $ and $\nu $, that is 
\begin{equation}
\gamma (A\times Y)=\mu (A)\quad \mathrm{and}\quad \gamma (X\times B)=\nu (B)
\label{marginals}
\end{equation}%
for all Borel $A\subset X,B\subset Y$. Intuitively, a matching is an
assignment of the agents in the sets $X$ and $Y$ into pairs, and $\gamma
\left( x,y\right) $ is related to the probability that $x$ will be matched
to $y$; in particular, $\left( x,y\right) \notin \spt \gamma$\footnote{%
Here $\mathop{\rm spt}\gamma $ refers to the support of $\gamma $, i.e. the
smallest closed set containing the full mass of $\gamma $.} implies that agents $x$
and $y$ are not matched together. The marginal condition is often called the 
\textit{market clearing} criterion. We denote the set of all matchings by $%
\Gamma (\mu ,\nu )$.

Integrable functions $u:X \rightarrow {\mathbf{R}}$ and $v:Y\rightarrow {%
\mathbf{R}}$ are called payoff functions corresponding to $\gamma$ if they
satisfy the \textit{budget constraint}:

\begin{equation}
u(x)+v(y)\leq s(x,y)  \label{budgetconstraint}
\end{equation}%
$\gamma $ almost everywhere --- i.e., for any pair of agents who match with
positive probability. For such a pair $(x,y)\in \mathop{\rm spt}\gamma $,%
 the functions $%
u(x)$ and $v(y)$ are interpreted respectively as the indirect utilities
derived from the match by agents $x$ and $y$; the constraint %
\eqref{budgetconstraint} ensures that the total indirect utility $u(x)+v(y)$
collected by the two agents does not exceed the total surplus $s(x,y)$
available to them.

A matching $\gamma $ is called \textit{stable} if there exist payoff
functions $u(x)$ and $v(y)$ satisfying both (\ref{budgetconstraint}) and the
reverse inequality 
\begin{equation}
u(x)+v(y)-s(x,y)\geq 0  \label{utilities}
\end{equation}%
%
%
%
%
%
%
%
%
%
%
%
%
%
%
%
%The measure $\gamma$  tells us which agents match together
for all $(x,y)\in X\times Y$. Condition \eqref{utilities} expresses the
stability of the matching in the following sense; if we had $%
u(x)+v(y)<s(x,y) $ for any (currently unmatched) pair of agents, it would be
desirable for each of them to leave their current partners and match
together, dividing the excess surplus $s(x,y)-u(x)-v(y)>0$ in such a way as
to increase the payoffs to both $x$ and $y$. Note that %
\eqref{budgetconstraint} and \eqref{utilities} together ensure $%
u(x)+v(y)=s(x,y)$, $\gamma $ almost everywhere: if two agents match with
positive probability, then they split the surplus generated between them.

For simplicity, we shall assume complete participation is incentivized
throughout and that supply balances demand (reflected by the fact that $\mu $
and $\nu $ have equal mass); when these assumptions are violated it is
well-known that they can be restored by augmenting both sides of the market
with a fictitious type representing the outside option of remaining
unmatched (see for instance \cite{cmn}).

Given a stable match $\gamma $ and associated matching functions $u,v$, the
set 
\begin{equation*}
S=\{(x,y)\in X\times Y\mid u(x)+v(y)=s(x,y)\}
\end{equation*}%
is of particular interest; as $\mathop{\rm spt}\gamma \subset S$, it tells
us which agents can match together. If $S$ is concentrated on a graph $\{(x,{%
F}(x))\mid x\in S\}$ of some function ${F}:X\longrightarrow Y$, the stable
matching is called \textit{pure}, the interpretation being that almost all
agents of type $x$ must match with agents of the same type $y={F}(x)$; in
particular, purity excludes the presence of \emph{randomization}, whereby an
agent may be randomly assigned to different partners. In this case, the
distribution $\nu $ agrees with the {image} ${F}_{\#}\mu $ of $\mu $ under ${%
F}$, which assigns mass %defined by
\begin{equation}
({F}_{\#}\mu )(V):=\mu \lbrack {F}^{-1}(V)]  \label{push-forward}
\end{equation}%
to each $V\subset Y$.\footnote{%
Also called the \emph{push-forward} ${F}_{\#}\mu $ of $\mu $ through ${F}$;
see e.g.\ \cite{akm}.} More generally, it is interesting and useful to
understand the geometry of~$S$.

Although $u$ and $v$ will not generally be everywhere differentiable, some
mild regularity condition guarantees differentiability almost everywhere, as
stated by the following result:

\begin{lemma}
\label{T:semiconvexity} If the surplus function $s$ is Lipschitz, so are the
payoffs $u$ and $v$ --- and with the same Lipschitz constant; if $s\in
C^{2}(X\times Y)$, then $u$ and $v$ have second-order Taylor expansions, 
%approximately differentiable twice,
Lebesgue almost-everywhere.
\end{lemma}

\begin{proof}
See \cite{V2} or \cite{Santambrogio15p}.
\end{proof}

When the probability measures $\mu $ and $\nu $ come from Lebesgue
densities, this almost-everywhere differentiability proves sufficient for
many analytic purposes. We use $\mathop{\rm Dom}Du$ (respectively $%
\mathop{\rm Dom}D^{2}u$) to denote those $x\in X$ at which $u$ has a first-
(respectively second-)order Taylor expansion, and $\mathop{\rm Dom}_0 D^iu
:= \left(\overline X\right)^0 \cap \mathop{\rm Dom} D^iu$ 
%$\Domo D^2u := (\overline X)^0 \cap \Dom D^2u$ 
where $\overline X$ and $X^0$ denote the closure and interior of $X$,
respectively.

The fact that $S$ is the zero-set of the non-negative function %
\eqref{utilities} enters crucially. It implies in particular the first-order
and second order conditions %\begin{eqnarray}
\begin{equation}
(Du(x),Dv(y))=(D_{x}s(x,y),D_{y}s(x,y))  \label{foc}
\end{equation}%
%
%
%
%
%\\ & D_y s(x,y) & = Dv(y)
%\end{equation}
and 
\begin{equation}
\left( 
\begin{matrix}
D^{2}u(x) & 0 \\ 
0 & D^{2}v(y)%
\end{matrix}%
\right) \geq \left( 
\begin{matrix}
D_{xx}^{2}s(x,y) & D_{xy}^{2}s(x,y) \\ 
D_{yx}^{2}s(x,y) & D_{yy}^{2}s(x,y)%
\end{matrix}%
\right)  \label{soc}
\end{equation}%
are satisfied at each $(x,y)\in S\cap (X\times Y)^{0}$ for which the
derivatives in question exist; here $X^{0}$ denotes the interior of $X$, and
inequality \eqref{soc} should be understood to mean the difference of these $%
(m+n)\times (m+n)$ symmetric matrices is non-negative definite.

The first-order condition 
\begin{equation}
Du(x)=D_{x}s(x,y)  \label{xfoc}
\end{equation}%
(for example) has an interesting, economic interpretation. In a
transferability utility model, the wife's share $u(x)$ of the surplus comes
at the expense of her husband $y$'s share. Thus \eqref{xfoc} expresses the
equality of his marginal willingness $Du(x)$ to pay for variations in her
qualities $x=(x^{1},\ldots ,x^{m})$ with the couple's marginal surplus $%
D_{x}s(x,y)$ for the same variations. Similarly 
\begin{equation}
Dv(y)=D_{y}s(x,y)  \label{yfoc}
\end{equation}%
equates her marginal willingness to pay for variations in his qualities $%
y=(y^{1},\ldots ,y^{n})$ with their marginal surplus for such variations.
This has important consequences for situations where characteristics are not
exogenously given but result from some investment made by individuals before
the beginning of the game (human capital being an obvious example). Then %
\eqref{xfoc} implies that\ the marginal return, for the individual, of an
investment in characteristics is exactly equal to the social return (defined
as the contribution of the investment to aggregate surplus). In other words,
one expect that for some equilibria such investments will be efficient,
despite being made non-cooperatively before the matching game; their impact
on global welfare is internalized by matching mechanisms, a point made by 
\cite{ColeMailathPostlewaite01} and \cite{IyigunWalsh07} and generalized by 
\cite{NoldekeSamuelson15p}.

\subsubsection{Variational interpretation: optimal transport and duality}

The problem of identifying stable matches turns out to have a variational
formulation, known as the optimal transport, or Monge-Kantorovich, problem
in the mathematics literature (see for instance \cite{V2}, \cite%
{Santambrogio15p} and \cite{Galichon16}). This is the problem of matching the measures $\mu $ and $%
\nu $ so as the maximize the total surplus; that is, to find $\gamma $ among
the set $\Gamma (\mu ,\nu )$ which maximizes

\begin{equation}  \label{MK}
s[\gamma]:=\int_{X \times Y}s(x,y)d\gamma(x,y).  \tag{MK}
\end{equation}

The following theorem can be traced back to \cite{SS} for finite type spaces 
$X$ and $Y$, and to \cite{GOZ} more generally. It asserts an equivalence
between \eqref{MK} and stable matchings.

\begin{theorem}
\label{equivalence} A matching measure $\gamma \in \Gamma(\mu,\nu)$ is
stable if and only if it maximizes \eqref{MK}.
\end{theorem}

As the maximization of a linear functional over a convex set, problem %
\eqref{MK} has a dual problem, which is useful both in studying it
maximizers, and in clarifying its relation with stable matching. The dual
problem to \eqref{MK} is to minimize 
\begin{equation}  \label{D}
\mu[u] +\nu[v]:=\int_{X }u(x)d\mu(x) +\int_{Y }v(y)d\nu(y).  \tag{MK$_*$}
\end{equation}
among functions $u \in L^1(\mu)$ and $v \in L^1(\nu)$ satisfying the
stability condition \eqref{utilities}. It is well known that under mild
conditions, duality holds (see, for instance, \cite{V2}), that is:

\begin{equation}  \label{duality}
\max_{\gamma\ \mathrm{satisfying}\ \eqref{marginals}} s[\gamma]=
\min_{(u,v)\ \mathrm{satisfying}\ \eqref{utilities}}\big(\mu[u] +\nu[v]\big).
\end{equation}
Note that for any $u$ and $v$ satisfying the stability constraint %
\eqref{utilities} and any matching $\gamma \in \Gamma(\mu,\nu)$, the
marginal condition implies

\begin{equation*}
\mu \lbrack u]+\nu \lbrack v]=\int_{X\times Y}\left( u(x)+v(y)\right)
d\gamma (x,y)\geq \int_{X\times Y}s(x,y)d\gamma (x,y)
\end{equation*}%
and we can have equality if and only if $u(x)+v(y)=s(x,y)$ holds $\gamma $%
-almost everywhere. It then follows from the duality theorem that $\gamma $
is a maximizer in \eqref{MK} (and hence a stable match) and $u,v$ are
minimizers in the dual problem \eqref{D}, precisely when $u(x)+v(y)=s(x,y)$
holds $\gamma $ almost everywhere; in other words, the solutions to \eqref{D}
coincide with the payoff functions.

An immediate corollary is the following:

\begin{corollary}
Let $s$ and $\bar{s}$ be two surplus functions. Assume there exists two
functions $f$ and $g$, mapping ${\mathbf{R}}^{m}$ to ${\mathbf{R}}$ and ${%
\mathbf{R}}^{n}$ to ${\mathbf{R}}$ respectively, such that 
\begin{equation*}
s\left( x,y\right) =\bar{s}\left( x,y\right) +f\left( x\right) +g\left(
y\right)
\end{equation*}%
For any measures $\mu $ and $\nu $, any stable matching for $s$ is a stable
matching for $\bar{s}$ and conversely.
\end{corollary}

\begin{proof}
Any stable measure $\gamma $ for $s$ solves the surplus maximization problem:%
\begin{equation}  \label{Prs}
\max_{\gamma\ \mathrm{satisfying}\ \eqref{marginals} }\int_{X\times
Y}s(x,y)d\gamma (x,y).
\end{equation}%
which is equivalent to:%
\begin{equation*}
\max_{\gamma\ \mathrm{satisfying}\ \eqref{marginals} }\int_{X\times Y}\bar{s}%
(x,y)d\gamma (x,y)+\int_{X}f(x)d\mu (x)+\int_{Y}g(y)d\nu (y)
\end{equation*}%
The last two integrals are given (by the marginal conditions on $\gamma $),
and so any $\gamma $ that solves (\ref{Prs}) also solves (\ref{Prsbar}): 
\begin{equation}  \label{Prsbar}
\max_{\gamma\ \mathrm{satisfying}\ \eqref{marginals} }\int_{X\times Y}\bar{s}%
(x,y)d\gamma (x,y).
\end{equation}
\end{proof}

An important consequence of this result is that\emph{\ the observation of
matching patterns can only (at best) identify the surplus up to two additive
functions of }$x$\emph{\ and }$y$\emph{\ respectively}. We shall see later
on that, in general, $s$ cannot be identified even up to two such additive
functions.

Problem \eqref{MK} has been studied extensively over the past 25 years, and
Theorem \ref{equivalence} allows the application of a large resulting body
of theory to the stable matching problem. In particular, conditions on $s$
leading to existence, uniqueness and purity of the solution to \eqref{MK}
are well known and surveyed in \cite{V2}. These properties can be compactly
and elegantly expressed in terms of the \emph{cross difference}, introduced
by \cite{Mc14} and defined on $(X\times Y)^{2}$ by:

\begin{equation*}
\delta(x,y,x_0,y_0) =s(x,y)+s(x_0,y_0)-s(x,y_0)-s(x_0,y)
\end{equation*}

The relevance of the cross difference to problem \eqref{MK} will become more
apparent in what follows. 
%; see e.g.\ Theorem \ref{T:characterizing optimality}. 
For now, we hint at its role by noting that in one dimension, $m=n=1$,
positivity of the cross difference whenever $(x-x_{0})\cdot (y-y_{0})>0$ is
equivalent to supermodularity of $s$ (the Spence-Mirrlees condition). For
more general $X$ and $Y$, the cross difference is zero along the diagonal $%
\{(x,y)=(x_{0},y_{0})\}$ and nonnegative along $\mathop{\rm spt}\gamma
\times \mathop{\rm spt}\gamma \subseteq (X\times Y)^{2}$ for any stable
match $\gamma $. This non-negativity of $s$ on $(\mathop{\rm spt}\gamma
)^{2} $ is typically referred to as \emph{$s$-monotonicity} of $%
\mathop{\rm
spt}\gamma $, and is a special case of a more general condition called \emph{%
$s$-cyclical monotonicity}, which characterizes the support of optimal
matchings.

\subsubsection{Existence, purity and uniqueness of a stable matching}

The variational formulation is quite helpful in establishing the basic
properties of stable matchings. Existence, for instance, can now be asserted
using basic continuity and compactness arguments, as stated by the following
result:

\begin{theorem}
\label{Xist}Assume $X\subset {\mathbf{R}}^{m}$ and $Y\subset {\mathbf{R}}%
^{n} $ are bounded and $s\in C(X\times Y)$. Then there exists an optimizer $%
\gamma $ to \eqref{MK}, and therefore a stable match.
\end{theorem}

\begin{proof}
See \cite{V2} or \cite{Santambrogio15p}
\end{proof}

We now consider uniqueness and purity. Aside from its theoretical interest,
uniqueness of the optimal matching $\gamma $ plays an important
computational role, as in its absence more sophisticated techniques must be
employed. In practice, solutions are often assumed to be pure in empirical
studies. Since this conclusion is not generically satisfied \cite%
{McCannRifford15p}, it is desirable to know conditions on $s$, $\mu $ and $%
\nu $ which guarantee it. % this property as well.
Furthermore, in hedonic contexts, purity implies the absence of 'bunching'
(whereby different agents consume the same product); this property will be
particularly important when we investigate the interactions between matching
and contract theory.

Note, first, that uniqueness is not guaranteed in general. For instance,
when the surplus function takes the additive form $s(x,y)=f(x)+g(y)$, the
functional

\begin{equation*}
s[\gamma]=\int_{X \times Y}s(x,y)d\gamma(x,y) =\int_Xf(x)d\mu(x)
+\int_Yg(y)d\nu(y)
\end{equation*}
is \textit{constant} throughout the set $\Gamma(\mu,\nu)$ of feasible
matching measures, and so any matching $\gamma$ is optimal and hence stable.
It is therefore clear that certain structural conditions on $s$ are indeed
needed to ensure purity and uniqueness.

The key condition for purity of the optimal matching is a nonlocal
generalization of the Spence-Mirrlees condition, %the following, 
known as the %\textit{generalized Spence-Mirrlees} or 
\textit{twist} condition:

\begin{definition}
The function $s \in C^1$ satisfies the \emph{twist condition }if, for each
fixed $x_{0}\in X$ and $y_{0}\neq y\in Y$, the mapping 
\begin{equation}  \label{twist}
x\in X\mapsto \delta (x,y,x_{0},y_{0})
\end{equation}%
has no critical points.
\end{definition}

By differentiating $\delta$ with respect to $x$ and rearranging, we see that
this is equivalent to

\begin{equation}  \label{twinjection}
D_{x}s(x,y)\neq D_{x}s(x,y_{0})
\end{equation}%
for all $x$ and distinct $y\neq y_{0}$. The twist 
%generalized Spence-Mirrlees 
condition is therefore equivalent to the \textit{injectivity} of $y\mapsto
D_{x}s(x,y)$, for each fixed $x$. This injectivity in turn implies that the
husband type $y$ of woman type $x\in \mathop{\rm Dom}Du$ is uniquely
determined by his marginal willingness $Du(x)$ to pay for her qualities
through the first-order condition~\eqref{xfoc}. For instance, in a
one-dimensional context ($m=1=n$), the classical Spence-Mirrlees condition
imposes that either $\frac{\partial ^{2}s}{\partial x\partial y}>0$ or $%
\frac{\partial ^{2}s}{\partial x\partial y}<0$ over $X\times Y$, which
implies that $y\mapsto \frac{\partial s}{\partial x}(x,y)$ is strictly
monotone (and hence injective) for each fixed $x$. It is in this sense that
the twist condition can be viewed as a non-local generalization of the
Spence-Mirrlees condition.

The twist condition is sufficient to guarantee purity, as stated by the
following result:

\begin{theorem}[\protect\cite{G}, \protect\cite{lev}]
\label{purity} Assume that $\mu $ is absolutely continuous with respect to
Lebesgue measure and the surplus $s$ satisfies the twist condition. Then any
solution $\gamma $ to \eqref{MK} is pure.
\end{theorem}

\begin{proof}
Let $\gamma$ solve \eqref{MK} for a twisted surplus $s \in C^1$. According
to Theorems \ref{T:semiconvexity} and \ref{equivalence}, there exist
Lipschitz potentials $u$ on $X$ and $v$ on $Y$ satisfying $u(x) + v(y) \ge
s(x,y)$ for all $(x,y) \in X \times Y$, with equality holding $\gamma$-a.e.
Being Lipschitz, $u$ is differentiable Lebesgue almost everywhere, hence $%
\mu $-a.e., where $\mu$ denotes the left marginal of $\gamma$. Thus for $\mu$%
-a.e.\ $x$ we find $Du(x) = D_x s(x,y)$ as in \eqref{foc}. The twist
condition \eqref{twist} implies we can invert this relation to write $y =
F(x)$, where $F(x) = [D_x s(x, \cdot)]^{-1}(Du(x))$. This shows $\gamma$ is
pure.
\end{proof}

%\begin{proof}
%See Appendix \ref{A:twist implies purity}. 
%\marginpar{
%BP: We don't currently have a proof of this...do we want to add one, or just
%provide a reference?
%\par
%\textbf{I would add a proof (unless it's really long) - PA}}
%\end{proof}

The absolute continuity of $\mu $ is a technical condition required to
ensure the utilities can be differentiated on a set of full $\mu $ measure;
the payoff functions are guaranteed to be Lipschitz if the surplus is, and
are hence differentiable Lebesgue almost everywhere by Rademacher's theorem
(but not everywhere in general). The condition on the measure can be
weakened somewhat, but some regularity is needed: as a simple
counterexample, if $\mu =\delta _{x_{0}}$ is a Dirac mass but $\nu $ is not,
then the optimal matching (indeed, the only measure in $\Gamma (\mu ,\nu )$)
is product measure $\delta _{x_{0}}\otimes \nu $, which pairs every point $y$
with $x_{0}$ and is certainly not pure.

Three remarks can be made at this point. First, if $s$ is twice continuously
differentiable and $Y$ has nonempty interior, the twist condition
immediately implies that $n\leq m$, as it asserts the existence of a smooth
injection \eqref{twinjection} from an open subset of ${\mathbf{R}}^{n}$ into 
${\mathbf{R}}^{m}$. Second, it is worth noting that in many relevant
situations, the twist condition does indeed fail; for example, if we replace 
$X$ and $Y$ with compact smooth manifolds, it fails for \textit{any} smooth
surplus function $s$. Third, the twist condition is \emph{not} necessary for
purity. For instance, \cite{KitagawaWarren12} provide a setting in which
purity holds in the absence of twist.

One can readily see that purity implies uniqueness:

\begin{corollary}
Under the conclusions of the preceding theorem, the optimal matching $\gamma 
$ is unique.
\end{corollary}

\begin{proof}
Suppose two solutions $\gamma_0$ and $\gamma_1$ to \eqref{MK} exist.
Convexity of the problem makes it clear that $\gamma_2=(\gamma_0+\gamma_1)/2$
is again a solution. The conclusion asserts that $\gamma_2$ concentrates on
the graph of a map ${F}:X\longrightarrow Y$, and vanishes outside this
graph. Non-negativity ensures the same must be true for $\gamma_0$ and $%
\gamma_1$. But then $\gamma_0= ({F} \times id)_\#\mu = \gamma_1$ by Lemma
3.1 of \cite{akm}
\end{proof}

The converse to this Corollary is not true; i.e., one can easily find
situations where the optimal matching is unique but not pure. Additional
conditions which ensure uniqueness, but not purity, of the optimal matching
can be found in \cite{cmn} and \cite{McCannRifford15p}.

\subsubsection{An example}

As an illustration of Theorem \ref{purity}, we consider a particular case
that has been widely used in empirical applications (\cite{GalichonSalanie12}%
, \cite{GalichonDupuy14}, \cite{Lindenlaub15} to name just a few). Assume
that $m=n$, and that the surplus takes the form:%
\begin{equation*}
s\left( x,y\right) = f_{X}\left( x\right) +g_{Y}\left( y\right) +
\sum_{k=1}^{n}f_{k}\left( x_{k}\right) g_{k}\left(y_{k}\right)
\end{equation*}

Two remarks can be made about this form. First, we may, without loss of
generality, disregard the first two terms by assuming that $f_{X}=g_{Y}=0$;
the stable measure will not be affected. Second, this form is necessary and
sufficient for the matrix of cross derivatives:%
\begin{equation*}
D_{xy}^{2}s=\left( \left( \frac{\partial ^{2}s}{\partial x_{k}\partial y_{l}}%
\right) \right)
\end{equation*}%
to be diagonal (a case investigated by \cite{Lindenlaub15}).

In that case, for any $y\neq \bar{y}$, we have that:%
\begin{equation*}
D_{x}s\left( x,y\right) -D_{x}s\left( x,\bar{y}\right) =\left( 
\begin{array}{c}
f_{1}^{\prime }\left( x_{1}\right) \left( g_{1}\left( y_{1}\right)
-g_{1}\left( \bar{y}_{1}\right) \right) \\ 
\vdots \\ 
f_{n}^{\prime }\left( x_{n}\right) \left( g_{n}\left( y_{n}\right)
-g_{n}\left( \bar{y}_{n}\right) \right)%
\end{array}%
\right)
\end{equation*}

In particular (and assuming that none of the $f_{i}$'s have anywhere
vanishing derivatives), the twist condition is satisfied whenever the $g_{i}$%
's are strictly monotonic. We conclude that, in that case, the stable match
is unique and pure; that is, there exists a function $F$ mapping ${\mathbf{R}%
}^{m}$ to itself, such that $x$ is matched a.s. with $y=\left( F_{1}\left(
x\right) ,...,F_{m}\left( x\right) \right) $. In addition, \cite%
{Lindenlaub15} show that if the $f$ and $g$ are strictly increasing, then $%
F_{i}$ is strictly increasing in $x_{i}$ for all $i$ (a property she calls
multidimensional positive assortative matching).

\subsubsection{Local structure of the optimal matching}

In the absence of the twist condition, or in the case where both marginals
are singular, one may not expect purity of the optimal match. However, under
a generic nondegeneracy criterion, which is a local generalization of the
Spence-Mirrlees condition, something can be still be asserted about the
local geometric structure of its support. For a fixed $(x_{0},y_{0})$, let $%
H $ be the Hessian of the function $(x,y)\mapsto \delta (x,y,x_{0},y_{0})$,
evaluated at the point $(x,y)=(x_{0},y_{0})$. 
%(or, equivalently, the Hessian of $s(x,y))$ in local coordinates,
A simple calculation yields that $H$ takes the block form: 
\begin{equation*}
H=%
\begin{bmatrix}
0 & D_{xy}^{2}s \\ 
D_{yx}^{2}s & 0%
\end{bmatrix}%
\end{equation*}%
where

\begin{equation*}
D_{xy}^{2}s=(\frac{\partial ^{2}s}{\partial x_{i}y_{j}})_{ij}
\end{equation*}%
is the $m\times n$ matrix of mixed second order partials, and the $0$'s in
the upper left and lower right hand corners represent $m\times m$ and $%
n\times n$ vanishing blocks, respectively. Denoting by $r$ the rank of $%
D_{xy}^{2}s$, we have the following theorem (see \cite{mpw} and \cite{P}).

\begin{theorem}
\label{spacelike} There exists a neighborhood $N$ of $(x_{0},y_{0})$ such
that the intersection 
% $N \cap  \spt\gamma$ of $N$ with the support of $\gamma$
$N\cap S$ is contained in a Lipschitz submanifold of $X\times Y$ of
dimension $m+n-r$. At points where $S$ is differentiable 
%submanifold of $X\timesY $ 
it is \emph{spacelike} in the sense that $H(v,v)\geq 0$ for any $v$ which is
tangent to the set $S$.
\end{theorem}

In particular, in the special case where $m=n$ and $H$ has full rank, $r=n$,
the support is at most $m=n$ dimensional (locally). When these dimensions
differ, and the $n\times m$ matrix has full rank $r=\min (n,m)$, the
dimension is $\max (n,m)$ and its codimension is $\min (n,m)$ \cite{P}. If $%
r=\min (n,m)$ we say $s$ is \emph{non-degenerate} at $(x_0,y_0)$.

The spacelike assertion of Theorem \ref{spacelike} can be useful to identify
the orientation of the optimal matching, as we will see below.

\begin{example}
When $m=n=1$, assuming the Spence-Mirrlees condition $\frac{\partial ^{2}s}{%
\partial x\partial y}>0$, the theorem tells us that the solution
concentrates on a $1$-dimensional Lipschitz curve $(x(t),y(t))$. Wherever
that curve is differentiable, the spacelike condition boils down to $%
x^{\prime }(t)\frac{\partial ^{2}s}{\partial x\partial y}y^{\prime }(t)\geq
0 $, or $x^{\prime }(t)y^{\prime }(t)\geq 0$, yielding positive assortative
matching. There is only one matching $\gamma $ with this structure, and it
can be computed explicitly from the measures $\mu $ and $\nu $ (provided $%
\mu $ has no atoms); it is given by ${F}_{\#}\mu $ where ${F}:X\rightarrow Y$
satisfies 
\begin{equation*}
\int_{(-\infty ,{F}(x))}d\nu \leq \int_{-\infty }^{x}d\mu \leq
\int_{(-\infty ,{F}(x)]}d\nu
\end{equation*}%
In words: $x$ is matched with $y=F\left( x\right) $ such that the number of
women whose characteristic is larger than $x$ equals the number of men above 
$y$.
\end{example}

\subsubsection{Regularity (smoothness) of the matching function}

\label{sect: regularity} When $s$ satisfies the twist condition, and the
optimal matching is therefore pure and unique, it is interesting to ask
whether the mapping ${F}:X\rightarrow Y$ generating the matching is
continuous; qualitatively, this is the question of whether women $x$ and $%
x^{\prime }$ whose types are `close' must marry men of similar
characteristics.

This turns out to generally not be the case, as a wide range of examples
throughout the literature on optimal transport show. In particular, as was
shown by \cite{mtw}, it is necessary (but not sufficient) that for each $%
x\in X$, the set $D_{x}s(x,Y)=\{D_{x}s(x,y)\mid y\in Y\}\subset {\mathbf{R}}%
^{n}$ be convex; if this condition fails, it is possible to construct
measures $\mu $ and $\nu $ with smooth positive densities for which ${F}$ is
discontinuous.

When $m>n$, this convexity is particularly restrictive; it can be shown to
fail, \textit{unless} the function $s$ takes a reducible, or \textit{index}
form, $s(x,y)=b(I(x),y)$, for some $I:X\rightarrow {\mathbf{R}}^{n}$ \cite%
{P2}. We will illustrate this phenomenon with an example later on.

Even when this convexity holds, the optimal map may not be continuous in
general; in this case, their regularity is governed by a restrictive fourth
order differential condition on $s$, known as the Ma-Trudinger-Wang
condition (see \cite{mtw} and \cite{loeper}). \cite{KimMcCann10} show this
condition to be equivalent to a sign condition on the curvature of the space 
$X\times Y$ geometrized using the Hessian $H$ as a (pseudo)-Riemannian
metric.\footnote{%
See %\citeasnoun{KimMcCann10}, and 
\cite{Mc14} for a general overview of the regularity of optimal mappings.}

\subsection{Extensions}

\subsubsection{Individual utilities}

\label{individual utilities}The stability condition allows information on
individual utilities at the stable match to be recovered. To see why, note
first that stability implies that, for $\mu $ almost every $x$,%
\begin{equation}
u\left( x\right) =\max_{y}\left( s\left( x,y\right) -v\left( y\right)
\right) .  \label{v transform}
\end{equation}

Assume, now, that the matching is pure (say, because the twist condition is
satisfied). The envelope theorem then yields, wherever $u$ is differentiable
and $y={F}(x)$ is matched with $x$,

\begin{equation}
\frac{\partial u}{\partial x_{i}}\left( x\right) =\frac{\partial s}{\partial
x_{i}}\left( x,{F}\left( x\right) \right) .  \label{envelope}
\end{equation}%
which gives the partials of $u$, and therefore defines $u$ up to an additive
constant. Note, incidentally, that these partial differential equations must
be compatible, which generates restrictions on the matching function ${F}$;
namely, assuming double differentiability:%
\begin{eqnarray}
\sum_{k}\frac{\partial ^{2}s}{\partial x_{i}\partial y_{k}}\frac{\partial {F}%
_{k}}{\partial x_{j}} &=&\sum_{k}\frac{\partial ^{2}s}{\partial
x_{j}\partial y_{k}}\frac{\partial {F}_{k}}{\partial x_{i}}  \label{CCond} \\
&\geq &\frac{\partial ^{2}s}{\partial x^{i}\partial x^{j}}(x,F(x))
\label{soc2}
\end{eqnarray}%
where $y_{k}={F}_{k}\left( x\right) $ and the partials of $s$ are taken at $%
\left( x,{F}\left( x\right) \right) $ and the inequality is from \eqref{soc}%
. In particular, in the case of multi to one dimensional matching, then $y={F%
}\left( x_{1},...,x_{m}\right) $, and (\ref{CCond}) becomes a system of
partial differential equations that ${F}$ must satisfy (which reduces to a
single equation in case $m=2=n+1$); together with the measure restrictions
and the matrix inequality \eqref{soc2}, this typically identifies the
matching function ${F}$. We shall come back to this point in Section 4.

\subsubsection{Index and pseudo-index models}

A very special case, which has been largely used in practical applications,
is obtained when the surplus function $s$ is weakly separable in one vector
of characteristics. Assume, indeed, that there exist two functions $I$ and $%
\sigma $, mapping ${\mathbf{R}}^{n}$ to ${\mathbf{R}}$ and ${\mathbf{R}}%
^{m+1}$ to ${\mathbf{R}}$ respectively, such that:%
\begin{equation}
s\left( y,x\right) =\sigma \left( x,I\left( y\right) \right) .  \label{Ind}
\end{equation}%
In words, the various male characteristics $y$ affect the matching function
only through some one dimensional index $I\left( y\right) $. This assumption
is quite restrictive; in practice, it requires that men with different
characteristics $y$ and $y^{\prime }$ but the same index (i.e., $I\left(
y\right) =I\left( y^{\prime }\right) $) be perfect substitutes on the
matching market \emph{for any potential spouse }$x$.

Formally, if $s$ is smooth, then the index form requires that $s$ satisfies
the following conditions:%
\begin{equation*}
\frac{\partial }{\partial x_{m}}\left( \frac{\partial s/\partial y_{k}}{%
\partial s/\partial y_{l}}\right) =0\ \ \forall k,l,m.
\end{equation*}%
These conditions express the fact that the marginal rate of substitution
between $y_{k}$ and $y_{l}$ (which defines the slope of tangent to the
corresponding iso-surplus curve) does not depend on $x$; indeed, (\ref{Ind})
implies that:%
\begin{equation*}
\frac{\partial s/\partial y_{k}}{\partial s/\partial y_{l}}=\frac{\partial
I/\partial y_{k}}{\partial I/\partial y_{l}}.
\end{equation*}

The main practical interest of index models is that, whenever (\ref{Ind}) is
satisfied, the matching problems is de facto one-dimensional in $y$;
technically, one can replace the space $Y$ and the measure $\nu $ with $%
\tilde{Y}=\text{Im}I\subset \mathbb{R}$ and the \emph{push-forward} $\tilde
\nu :=I_\#\nu$ of $\nu $ through $I$ defined as in \eqref{push-forward}. 
%\footnote{DELETE DUE TO REDUNDANCY WITH \eqref{push-forward}:
%The \emph{push-forward} of $\nu $ through $I$ is the measure $I_{\#}\nu $
%induced on $\text{Im}I$ by the mapping $\nu $. Technically, it is defined on 
%$\text{Im}I$ by the formula $I_{\#}\nu (B)=\nu \left( I^{-1}\left( B\right)
%\right) $ for any Borel $B\subset \text{Im}I$.} 
In particular, when the index property (\ref{Ind}) is satisfied, then the
matching problem boils down to a multi-to-one dimensional problem, of the
type discussed in Section 4.\footnote{%
A practical difficulty is that, for most empirical applications, the index $%
I $ is not known ex ante and has to be empirically estimated. See \cite{COQ}
for a precise discussion.}

Another notable property of index models, that was observed in \cite{P2}, is
that when $m\leq n$, the convexity of the sets $\{D_{y}s(y,x)\mid x\in X\}$,
necessary to ensure that the matching function is continuous, \emph{requires}
that the function $s$ takes the form $s(x,y)=\sigma (x,I(y))$, where here $I:%
\mathbb{R}^{n}\rightarrow \mathbb{R}^{m}$. In particular, when $m=1$, $s$
must have an index form. Therefore, when $n>m=1$, if a surplus function $s$
is not of index form, there are absolutely continuous measures $\mu $ and $%
\nu $ (with smooth densities) for which the matching function ${F}$ is
discontinuous; we will see an example of this below.

Lastly, the notion of index model can for some purposes be slightly relaxed.
Specifically, we define a \emph{pseudo-index }model by assuming that there
exist three functions $\alpha $, $I$ and $\sigma $, mapping ${\mathbf{R}}%
^{n} $ to ${\mathbf{R}}$, ${\mathbf{R}}^{n}$ to ${\mathbf{R}}$ and ${\mathbf{%
R}}^{m+1}$ to ${\mathbf{R}}$ respectively, such that:%
\begin{equation}
s\left( x,y\right) =\alpha \left( y\right) +\sigma \left( x,I\left( y\right)
\right) .
\end{equation}%
Here, male characteristics $y$ affect the matching function through \emph{two%
} one-dimensional indices $\alpha \left( y\right) $ and $I\left( y\right) $.
The crucial remark, however, is the following. Assume that $D_{x}\sigma
\left( x,i\right) $ is injective in $i$; this simply requires that $\partial
\sigma /\partial x_{k}\left( x,i\right) $ is strictly monotonic in $i$ for
at least one $k$. Then:

\begin{equation*}
D_{x}s(x,y)=D_{x}\sigma \left( x,I\left( y\right) \right) \neq D_{x}\sigma
\left( x,I\left( y_{0}\right) \right) =D_{x}s(x,y_{0})
\end{equation*}%
for any $y,y_{0}$ such that $I\left( y\right) \neq I\left( y_{0}\right) $.
It follows from the proof of the twist theorem that the support of the
stable measure is born by the graph of a function. Although the stable
matching on $\mathbf{R}^{m}\times \mathbf{R}^{n}$ is not pure (since all
males with the same index are perfect substitutes), for the `reduced'
matching problem defined on $\mathbf{R}^{m}\times \mathbf{R}$ by the surplus
function $\sigma \left( x,i\right) $ and the measures $\mu $ and $I_{\#}\nu $%
, the stable matching is pure - i.e., there exists a function $\phi $ such
that any women $x$ is matched with probability one with a man whose index is 
$i=\phi \left( x\right) $. In particular, most results derived in the multi
to one dimensional case (see Section 4 below) still apply in that case. 
%\marginpar{RM: why distinguish $f$ from $F$ from $\phi$?}

\subsubsection{Links with hedonic models}

Next, we briefly recall the canonical link between matching and hedonic
models,\footnote{%
See \cite{cmn}, \cite{ekeland2} and \cite{ekeland3}.} which will be crucial
for some of our applications (particularly the competitive IO model and the
competitive version of Rochet-Chon\'{e}). An hedonic model involves three
sets: a set $X$ of buyers (endowed with a measure $\mu $), a set $Y$ of
sellers (endowed with a measure $\nu $) and a set $Z$ of products;
intuitively, a product is defined by a finite vector of characteristics, and
buyers purchase one product (at most) in any period --- think of a car or a
house, for instance. Both buyers and sellers are price-takers, and consider
the price $P\left( z\right) $ of any product $z\in Z$ as given. A buyer with
a vector of characteristics $x$ maximizes a quasi-linear utility of the form%
\begin{equation*}
U=u\left( x,z\right) -P\left( z\right)
\end{equation*}%
while producer $y$ maximizes profit%
\begin{equation*}
\Pi =P\left( z\right) -c\left( y,z\right)
\end{equation*}%
where $c\left( y,z\right) $ is the cost, for producer $y$ to produce product 
$z$.

\begin{remark}
It is important to note that the producer's profit depends on his
characteristics, on the product's characteristics and on the price, but 
\emph{not }on the characteristics of the buyer. This `private value' aspect
will be crucial for the relationship with bidimensional adverse selection
that we discuss below.
\end{remark}

An hedonic equilibrium is defined as a measure $\alpha $ on the product set $%
X\times Y\times Z$, whose first and second marginals are $\mu $ and $\nu $,
and a function $P$ such that, for any $\left( \bar{x},\bar{y},\bar{z}\right)
\in \mathop{\rm
spt}\alpha $, then 
\begin{equation*}
\bar{z}\in \arg \max_{z}\left( u\left( x,z\right) -P\left( z\right) \right)
\cap \arg \max_{z}\left( P\left( z\right) -c\left( y,z\right) \right).
\end{equation*}%
In words, $\alpha $ represents an assignment of buyers and sellers to each
other and to products, and an equilibrium is reached if the product assigned
to $x$ (resp. $y$) maximizes $x$'s utility ($y$'s profit).

To see the link with matching models, define the pairwise surplus function 
\begin{equation}
s\left( x,y\right) =\sup_{z\in Z}u\left( x,z\right) -c\left( y,z\right) .
\label{surplus}
\end{equation}%
and consider the matching model defined by $\left(X,Y,s\right) $. Then \cite%
{cmn} prove the following results:

\begin{itemize}
\item for any equilibrium of the hedonic model, the projection of the
measure $\alpha $ on the set $X\times Y$, together with the functions%
\begin{equation*}
U\left( x\right) =\max_{z}u\left( x,z\right) -P\left( z\right) \ \text{and }%
V\left( y\right) =\max_{z}\left( P\left( z\right) -c\left( y,z\right) \right)
\end{equation*}%
form a stable matching.

\item Conversely, for any $\gamma $ that solves the matching problem, there
exist a price function $P$ that satisfies 
\begin{equation}
\inf_{y\in Y}\left\{ v\left( y,z\right) +r\left( y\right) \right\} \geq
P\left( z\right) \geq \sup_{x\in X}\left\{ u\left( x,z\right) -q\left(
x\right) \right\}  \label{squeeze}
\end{equation}%
With $\alpha \equiv \left( id_{X}\times id_{Y}\times z_{0}\right)
_{\#}\gamma $, where $z_{0}=z_{0}(x,y)\in \arg \max_{z}u(x,z)-c(y,z)$, any
such $P$ forms an equilibrium pair $\left( \alpha ,P\right) $.
\end{itemize}

\subsubsection{Testability}

Lastly, let's briefly consider the important issue of testability: what
testable restrictions (if any) are generated by the general matching
structure described above? Obviously, the answer depends on what we can
observe. Consider the simplest case, in which we only observe matching
patterns (`who marries whom'). Technically, we are now facing an inverse
problem: knowing the spaces $X$ and $Y$ and the measure $\gamma $, can we
find a surplus $s$ for which $\gamma $ is stable? This question should
however be slightly rephrased to rule out degerate solutions; for instance,
any measure is stable for the degenerate surplus $s\left( x,y\right) =0\ \
\forall x,y$.

We therefore consider the following problem: Given two spaces $X$, $Y$ and
some measure $\gamma $ on $X\times Y$, is is always possible to find a
surplus $s$ such that $\gamma $ is the \emph{unique }stable matching of the
matching problem $\left( X,Y,s\right) $?

A first remark is that if we impose enough `regularity' in the model, the
answer is positive. Specifically, let us consider the case in which the
support of the measure is born by the graph of some function $F$, and that $%
F $ is non-degenerate (in the sense that the derivative of $F$ has full rank
over the entire space). Then one can always find a surplus for which $\gamma 
$ is the unique stable matching: we just need to take $s(x,y)=-|{F}%
(x)-y|^{2}/2$. Indeed, $\gamma $ obviously maximizes the primal, optimal
transportation problem, which guarantees stability; moreover, the surplus
satisfies the twist condition, which guarantees uniqueness. The
corresponding payoffs are $u(x)=0=v(y)$.

However, non-degeneracy is crucial for this result to hold. For one thing,
if $F$ is degenerate, the twist condition does not hold, and while $\gamma $
is always stable for $s\left( x,y\right) =-|{F}(x)-y|^{2}/2$, it may not be
the unique stable matching. Moreover, and while Theorem \ref{spacelike}
implies that any stable matching for a non-degenerate surplus concentrates
on a set of dimension at most $\max (m,n)$, it is possible to find measures
supported on sets of this dimension which are \emph{not} stable for any $%
C^{2}$, \emph{non-degenerate} surplus. To see this, consider the $m=n=1$
case; let $X=Y=(0,1)\subseteq \mathbf{R}$. Nondegeneracy here simply means $%
\frac{\partial ^{2}s}{\partial x\partial y}\neq 0$, which implies either $%
\frac{\partial ^{2}s}{\partial x\partial y}>0$ everywhere (so $s$ is
super-modular) or $\frac{\partial ^{2}s}{\partial x\partial y}<0$ everywhere
(so $s$ is submodular). In these two cases, it is well known that stable
matches concentrate on monotone increasing or decreasing sets, respectively.
Therefore, any $\gamma $ concentrating on a set of dimension $\max (m,n)=1$
(for instance, a smooth curve), which is neither globally increasing nor
decreasing (for example, the curve ($y=4(x-1/2)^{2}$), cannot be stable for
any non-degenerate surplus.

\section{Matching with unequal dimensions}

We now turn special attention to the case in which the dimensions $m\geq n$
of heterogeneity on the two sides of the market are unequal. In this case,
one expects many-to-one rather than one-to-one matching. In a companion
publication, we develop a detailed mathematical theory for this situation,
focusing %specifically on the 
especially on the case $m > n=1$. Here we announce only the main conclusions
of that theory and a few of the underlying ideas, suppressing technical
details wherever possible to be able to move quickly on to interpretation
and applications.

When $m\ge n$, it is natural to expect that at equilibrium the subset ${F}%
^{-1}(y)\subset X \subset {\mathbf{R}}^m$ of partners which a man of type $%
y\in \mathop{\rm Dom}_0 Dv$ is indifferent to will generically have
dimension $m-n$, or equivalently, codimension $n$. We are interested in
specifying conditions under which this indifference set will in fact be a
smooth submanifold.\footnote{%
In the familiar case $n=m$ it would then consist of one (or more) isolated
points --- a single one if $s$ happens to be twisted.} Let us explore this
situation, paying particular attention to the separate roles played by the
surplus function $s(x,y)$ as opposed to the populations $\mu $ and $\nu $ in
determining these indifference sets.

\subsection{Potential indifference sets}

For any equilibrium matching $\gamma$ and payoffs $(u,v)$, we have already
seen that $(x,y) \in S \cap (X\times \mathop{\rm Dom}_0 Dv)$ implies %
\eqref{yfoc}. That is, all partner types $x \in X$ for husband $y \in %
\mathop{\rm Dom}_0 Dv$ lie in the same level set of the map $x \mapsto D_y
s(x,y)$. If we know $Dv(y)$, we can determine this level set precisely; it
depends on $\mu$ and $\nu$ as well as $s$. However, in the absence of this
knowledge it is useful to define the \emph{potential indifference sets},
which for given $y \in Y$ are merely the level sets of the map $x\in X
\mapsto D_y s(x,y)$. We can parameterize these level sets by (cotangent)
vectors $k \in {\mathbf{R}}^n$: 
\begin{equation}  \label{cotangent-parameterization}
X(y,k) := \{ x \in X \mid D_y s(x,y) = k\},
\end{equation}
or we can think of $y \in Y$ as inducing an equivalence relation between
points of $X$, under which $x$ and $\bar x \in X$ are equivalent if and only
if 
\begin{equation*}
D_y s(x,y) = D_y s(\bar x,y).
\end{equation*}%
%
%
%
%
%
%
%
%
%
%
%
%
%
%
%
%\end{equation}
Under this equivalence relation, the equivalent classes take the form %
\eqref{cotangent-parameterization}. We call these equivalence classes \emph{%
potential indifference sets}, since they represent a set of partner types
which $y \in \mathop{\rm Dom} Dv$ has the potential to be indifferent
between. The equivalence class containing a given partner type $\bar x \in X$
will also be denoted by 
\begin{equation}  \label{equivalence class}
L_{\bar x}(y) = X(y,D_y s(\bar x,y)) = \{ x \in X \mid D_y s(x,y) = D_y
s(\bar x,y)\}.
\end{equation}

A key observation concerning potential indifference sets is the following
proposition. % consequence of the inverse function theorem.

\begin{definition}[Surplus degeneracy]
Given $X \subset {\mathbf{R}}^m$ and $Y \subset {\mathbf{R}}^n$, we say $s
\in C^2(X\times Y)$ \emph{degenerates} at $(x,y) \in X \times Y$ if $%
rank(D^2_{xy}s(\bar x,\bar y))<\min\{m,n\}$. Otherwise we say $s$ is \emph{%
non-degenerate} at $(\bar x,\bar y)$.
\end{definition}

\begin{proposition}[Structure of potential indifference sets]
\label{P:indifference structure} Let $s \in C^{r+1}(X \times Y)$ for some $%
r\ge 1$, where $X \subset {\mathbf{R}}^m$ and $Y \subset {\mathbf{R}}^n$
with $m \ge n$. If $s$ does not degenerate at $(\bar x,\bar y) \in X \times
Y $, then $\bar x$ admits a neighbourhood $U \subset {\mathbf{R}}^m$ such
that $L_{\bar x}(\bar y) \cap U$ coincides with the intersection of $X$ with
a $C^r$-smooth, codimension $n$ submanifold of $U$.
\end{proposition}

\begin{proof}
Since $s \in C^2$, the surplus extends to a neighbourhood $U \times V$ of $%
(\bar x,\bar y)$ on which 
%the maximal rank condition $rank(D^2_{xy}s(x,y))=\min\{n,m\}$ 
$s$ continues to be non-degenerate (by lower semicontinuity of the rank).
The set $\{ x \in U \mid D_y s(x,y) = D_y s(\bar x,\bar y) \}$ forms a
codimension $n$ submanifold of $U$, by the preimage theorem \cite[\S 1.4]%
{GuilleminPollack74}. More specifically, the rank condition implies that
choosing a suitable orthonormal basis for ${\mathbf{R}}^m$ yields $\det [%
\frac{\partial^2 s}{\partial x^i \partial y^j}(\bar x,\bar y)]_{1\le i,j\le
n} \ne 0$. In these coordinates, the potential indifference set is locally
parameterized as the inverse image under the $C^r$ map $x \in U \mapsto (D_y
s(x,\bar y),x_{n+1},\ldots,x_m)$ of the affine subspace $\{D_y s(\bar x,\bar
y)\} \times {\mathbf{R}}^{n-m}$. Taking $U$ and $V$ smaller if necessary,
the inverse function theorem then shows $L_{\bar x}(\bar y) \cap U $ to be $%
C^r$.
\end{proof}

Although we have stated the proposition in local form, it implies that if $%
\bar{k}=D_{y}s(\bar{x},\bar{y})$ is a \emph{regular value} of $x\in X\mapsto
D_{y}s(x,\bar{y})$ --- meaning $D_{xy}^{2}s(x,\bar{y})$ has rank $n$
throughout $L_{\bar{x}}(\bar{y})$ --- then $L_{\bar{x}}(\bar{y})=X(\bar{y},%
\bar{k})$ is the intersection of $X$ with an $m-n$ dimensional submanifold
of ${\mathbf{R}}^{m}$. Note however that this proposition says nothing about
points $(\bar{x},\bar{y})$ where $s$ degenerates, which can happen
throughout $\mathop{\rm
spt}\gamma $.

\subsection{Potential versus actual indifferences sets}

As argued above, the potential indifference sets %
\eqref{cotangent-parameterization} and \eqref{equivalence class} are
determined by the surplus function $s(x,y)$ without reference to the
populations $\mu $ and $\nu $ to be matched. On the other hand, the
indifference set actually realized by each $y\in Y$ depends on the
relationship between $\mu $, $\nu $ and $s$. This dependency is generally
complicated, as illustrated by the following example.

\begin{example}
Consider the surplus function:%
\begin{equation*}
s\left( x,y\right) =x_{1}y_{1}+x_{2}y_{2}+x_{3}y_{1}y_{2}
\end{equation*}%
where $X\subset {\mathbf{R}}^{3},Y\subset {\mathbf{R}}^{2}$. The potential
indifference sets are given, for any $k\in {{\mathbf{R}}}^{2}$, by: 
\begin{equation}
X(y,k):=\left\{ x\in X\mid 
\begin{array}{c}
x_{1}+x_{3}y_{2}=k_{1}\text{ } \\ 
\text{and} \\ 
x_{2}+x_{3}y_{1}=k_{2}%
\end{array}%
\right\} .
\end{equation}%
These are straight lines in ${\mathbf{R}}^{3}$, parallel to the vector $%
\left( 
\begin{array}{c}
y_{2}\text{ } \\ 
y_{1} \\ 
-1%
\end{array}%
\right) $. Therefore, for any given $y\in {\mathbf{R}}^{2}$, we know that
the set of spouses matched with $y$ (the indifference set corresponding to
husband $y$) will be contained in such a straight line. However, it is certainly not true
that any such line (obtained for an arbitrary choice of $k$) will be an
indifference set curve. For a given $y$, the exact equation of the
indifference set corresponding to $y$ is defined by the value of the
specific vector $k$ which is relevant for that particular $y$ --- and this
depends on the measures $\mu $ and $\nu $.
\end{example}

However, there is one case in the problem may simplify substantially: the
case of multi-to-one dimensional matching, namely $n=1$. In this case,
suppose $D_{xy}^{2}s(\cdot ,y)$ is non-vanishing (i.e.\ $\frac{\partial s}{%
\partial y}(\cdot ,y)$ takes only regular values). %say $D_y s(\cdot,y)>0$.
Then the potential indifference sets $X(y,k)$ are codimension $1$ in ${%
\mathbf{R}}^{m}$; that is, they are curves in ${\mathbf{R}}^{2}$, surfaces
in ${\mathbf{R}}^{3}$, and hypersurfaces in higher dimensions $m\geq 4$.
Moreover, as $k$ moves through ${\mathbf{R}}$, these potential indifference
sets sweep out more and more of the mass of $\mu $. For each $y\in Y$ there
will be some choice of $k\in {\mathbf{R}}$ for which the $\mu $ measure of $%
\{x\mid D_{y}s(x,y)\leq k\}$ exactly coincides with the $\nu $ measure of $%
(-\infty ,y]$ (assuming both measures are absolutely continuous with respect
to Lebesgue, or at least that $\mu $ concentrates no mass on hypersurfaces
and $\nu $ has no atoms). In this case the potential indifference set $%
X(y,k) $ is said to split the population proportionately at $y$, making it a
natural candidate for being the true indifference set ${F}^{-1}(y)$ to be
matched with $y$.\footnote{%
Since $k=s_{y}(x,y)$ can be recovered from any $x\in X(y,k)$ and $y$, we may
equivalently say $x$ splits the population proportionately at $y$, and vice
versa.} In the next Section, we go on to describe and contrast situations in
which this expectation is born out and leads to a complete solution from
those in which it does not.

\section{Multi-to-one dimensional matching}

We now explain a new approach to a specific class of models, largely
unexplored in either the mathematics or economics literature, but which can
often be solved explicitly with the techniques outlined below and developed
more fully in \cite{ChiapporiMcCannPass15m}. These are \emph{multi-to-one
dimensional} models, in which agents on one side of the market (say wives)
are bi-dimensional (or, potentially, higher dimensional) while agents on the
other side (husbands) are one-dimensional. Thus, we are matching a
distribution on $x=\left( x_{1},...,x_{m}\right) \in {\mathbf{R}}^{m}$ with
another on $y\in {\mathbf{R}}$. The surplus $s$ is then a function $s\left(
x_{1},...,x_{m},y\right) $ of $m+1$ real variables, and will typically be
increasing in each argument.

The crucial notion, is this setting, is that of a \emph{iso-husband curve},
defined as the indifference set of a given husband $y$, i.e. as the
submanifold of wives among which husband $y$ turns out to be indifferent
facing the given market conditions. Iso-husband curves, as we shall see,
play a key role in the construction of an explicit solution to the
matching/optimal transportation problem. In addition, a crucial property of
these curves is that they can in principle be empirically identified; see 
\cite{COQ} for a detailed discussion. In fact, it has been argued that the
theoretical properties of iso-husband curves could provide the most powerful
empirical tests of matching theory (see for instance \cite{COQ1}).

The goal is to construct from data $(s,\mu ,\nu )$ a matching function ${F}%
:X\longrightarrow Y\subset {\mathbf{R}}$, whose level sets ${F}^{-1}(y)$
constitute iso-husband curves. At the end of the preceding section we
identified a natural candidate for this indifference set: namely the
potential indifference set which divides the mass of $\mu $ in the same
ratio as $y$ divides $\nu $; whether or not these natural candidates
actually fit together to form the level sets of a function or not depends on
a subtle interaction between $\mu $, $\nu $ and $s$. When they do, we say
the model is \emph{nested}, and in that case we show that the resulting
function ${F}:X\longrightarrow Y$ produces a stable equilibrium match.

%\marginpar{RM: The PDE strikes me as a red-herring which should be de-emphasized}
%We first show how these matching models can be formulated as a partial
%differential equation in the variables $x$, and describe an algorithm which
%quite often yields the explicit solution. 

Mathematically, the simplest multi-to-one dimensional models arise from
index surpluses, which have economic motivation arising from \cite{COQ}.
These models are nested for \emph{every} choice of the distributions $\nu $
and $\mu $, as we demonstrate below. Moreover, in the following section, we
will discuss three applications of multidimensional matching theory; the
first two of these deal explicitly with multi-to-one dimensional problems.
The first model arises in the marriage market, where recent research
indicates that it is appropriate to model women using both education and
fertility and men using income only \cite{L}\cite{L2}. A second example
comes from a hedonic variant of the Rochet-Chon\'{e} screening problem.

%We then go on to discuss several explicit examples; the first, and simplest
%is the case an index cost, with economic motivation arising from \cite{COQ}.
%A second model arises in the marriage market, where recent research
%indicates that it is appropriate to model women using both education and
%fertility and men using income only \cite{L}\cite{L2}. A third example comes
%from a hedonic variant of the Rochet-Chon\'{e} screening problem.

\subsection{Constructing explicit solutions for nested data}

We now give a heuristic description of a general algorithm for constructing
the solution to the matching problem when one of the dimensions is $n=1$;
the aim is to find a solution to \eqref{CCond} above (or,
equivalently, \eqref{envelope}), which satisfies the mass balance condition $%
\nu ={F}_{\#}\mu $ and the spacelike condition in Theorem \ref{spacelike}.
In order to work, our approach requires the nestedness property mentioned
above and detailed below. These conditions are satisfied in a wide class of
multi-to-one dimensional matching problems; they are illustrated in the
theorem and examples presented below. However, except in the Spence-Mirrlees
(with $m=n=1$) and in the index and pseudo-index cases, this nestedness
depends not only on $s$, but also on $\mu $ and $\nu $.

For each fixed $y\in Y\subseteq {\mathbf{R}}$, our goal is to identify the
iso-husband (or indifference) set $\{x\in X\mid {F}(x)=y\}$ of husband type $%
y$ facing the given market conditions. When differentiability of $v$ holds
at $y$, the argument in the preceding section implies that this is contained
in one of the potential indifference sets $X(y,k)$ from %
\eqref{cotangent-parameterization}. 
%:=\{x\in X\mid \frac{\partial s}{\partial y}(x,y)=k\}$. 
%given by level sets of the function $x \mapsto \frac{\partial s}{\partial y}(x,y)$.  
Proposition \ref{P:indifference structure} indicates when this set will have
codimension $1$; it generally divides $X$ into two pieces: the sublevel set 
\begin{equation}  \label{Xsubseteq}
X_{\leq }(y,k):=\{x\in X\mid \frac{\partial s}{\partial y}(x,y)\leq k\},
\end{equation}%
and its complement $X_{>}(y,k):=X\setminus X_{\leq }(y,k)$. We denote its
strict variant by $X_{<}(y,k):=X_{\leq }(y,k)\setminus X(y,k)$.

To select the appropriate level set, 
% curve, $\{x\mid \frac{\partial s}{\partial y}(x,y)=k\}$, 
assuming ${F}$ is differentiable, the spacelike condition implies $\frac{%
\partial {F}}{\partial x_i} (x) \frac{\partial^2 s}{\partial x_i\partial y}
\geq 0$ for each $i=1,\ldots, m$. 
% hence $\nabla_xF(x)\cdot \nabla \frac{\partial s}{\partial y} \geq 0$, and so we expect 
Summing on $i$ suggests points $\bar x$ in the region $X_\le(y,k)$ 
%$\frac{\partial s}{%\partial y}(\bar x,y) > k$ 
get paired with points $\bar y \le y$. We therefore choose the unique level
set \textit{splitting the population proportionately} with $y$; that is, the 
$k=k(y)$ for which the $\mu$ measure of female types $X_\le(y,k)$ 
%\{x\mid \frac{\partial s}{\partial y}(x,y) > k\}$ 
coincides with the $\nu$ measure of male types $(-\infty,y]$. We then set $%
y:={F}(x)$ for each $x$ in the set $X(y,k)$. 
%$\{x\mid \frac{\partial s}{\partial y}(x,y) = k\}$.

Our first theorem specifies conditions under which the resulting match $%
\gamma = (id \times {F})_\#\mu$ optimizes the Kantorovich problem \eqref{MK}%
; we view it as the natural generalization of the positive assortative
matching results of \cite{Mirrlees71} \cite{Becker73} and \cite{Spence73}
from the one-dimensional to the multi-to-one dimensional setting. Unlike
their criterion, which depends only on $s$, ours relates $s$ to $\mu$ and $%
\nu$, by requiring the sublevel sets $y \in Y \mapsto X_\le(y,k(y))$
identified by the procedure above to depend monotonically on $y\in {\mathbf{R%
}}$, with the strict inclusion $X_\le(y,k(y))\subset X_<(y^{\prime
},k(y^{\prime }))$ holding whenever $\nu[(y,y^{\prime })]>0$. We say the
model $(s,\mu,\nu)$ is \emph{nested} in this case. Our situation is
naturally more complicated than theirs, 
%that of Becker, Mirrlees and Spence, 
since there is no obvious ordering of the women's types, but generally a
variety of possible orderings depending on population frequencies $\mu$ and $%
\nu$; nestedness rather asserts that the women's preferences enjoy some
degree of compatibility, in the sense that for %each %$
$[\underline y,\bar y] \subset Y$ any wife $\bar x \in X(\bar y,k(\bar y))$
assigned to the higher type of husband has a greater willingness to pay for
variations in the qualities of husband type $\bar y$ (and similarly of
husband type $\underline y$) than does any wife $\underline x \in
X(\underline y,k(\underline y))$ assigned to the lower type of husband.

%$y^+$ or $y^-$.

%have a greater marginal willingness to pay for variations in the quality of husband $y^+$
%than for variations in the quality of husband $y^-<y+$.

%to pay for the qualities of the wife assigned to $y^+>y$ exceed his marginal willingness to pay for the qualities %of the wife  assigned to 
%$y^-<y$.

\begin{theorem}[Optimality of nested matchings]
\label{T:nested} Let $X \subset {\mathbf{R}}^m$ and $Y\subset {\mathbf{R}}$
be connected open sets equipped with Borel probability measures $\mu$ and $%
\nu$. Assume $\nu$ has no atoms and $\mu$ vanishes on each $C^1$
hypersurface. Use $s \in C^2(X\times Y)$ and $s_y =\frac{\partial s}{%
\partial y}$ to define $X_\le$, $X_<$ etc., as in \eqref{Xsubseteq}

Assume $s$ is non-degenerate, $|D_x s_y| \ne 0$, throughout ${X \times Y}$.
Then for each $y \in \overline Y$ there is a maximal interval $%
K(y)=[k^-(y),k^+(y)] \ne \emptyset$ such that $\mu[X_\le(y,k)] =\nu[%
(-\infty,y)]$ for all $k\in K(y)$. Both $k^+$ and $-k^-$ are upper
semicontinuous. Assume both maps $y \in Y \mapsto X_\le (y,k^\pm(y))$ are
non-decreasing, and moreover that $\int_y^{y^{\prime }} d\nu >0$ implies $%
X_\le (y,k^+(y)) \subseteq X_< (y^{\prime -}(y^{\prime }))$. Then $k^+=k^-$
holds $\nu$-a.e. Setting ${F}(x) = y$ for each $x \in X(y,k^+(y))$ defines a
stable match ${F}:X \longrightarrow Y$ [$\mu$-a.e.]. Moreover, $\gamma=(id
\times {F})_\#\mu$ maximizes \eqref{MK} uniquely on $\Gamma(\mu,\nu)$.
Finally, if $\mathop{\rm spt} \nu$ is connected then ${F}$ extends
continuously to $X$.
\end{theorem}

\begin{proof}[Idea of proof]
Non-degeneracy implies $X(y,k) := X_\le(y,k) \setminus X_<(y,k)$ is an $m-1$
dimensional $C^1$ submanifold of $X$ orthogonal to $D_x s_y(x,y) \ne 0$.
Since both $\mu$ and $\nu$ vanish on hypersurfaces, the function 
\begin{equation}  \label{h}
h(y,k):=\mu[X_\le(y,k)]-\nu[(-\infty,y)]
\end{equation}
is continuous, and for each $y \in Y$ climbs monotonically from $-\nu[%
(-\infty,y)]$ to $1- \nu[(-\infty,y)]$ with $k \in {\mathbf{R}}$. This
proves the existence of $k^\pm(y)$ and confirms the zero set of $h(y,k)$ is
closed. Thus $k^-$ is lower semicontinuous, $k^+$ is upper semicontinuous,
and by the intermediate value theorem $[k^-(y),k^+(y)]$ is non-empty.

The main strategy for the rest of the proof is to use $k^+(y)$ to construct
a Lipschitz equilibrium payoff function $v$ by solving $v^{\prime +}(y)=k^+(y)$
a.e. Together with $u$ from \eqref{v transform}, it can be shown that $(u,v)$
minimizes the dual problem \eqref{D} and $\gamma = (id \times {F})_\# \mu$
maximizes the planners problem \eqref{MK}. For details and the case $s \in
C^{1,1}$, see \cite{ChiapporiMcCannPass15m}.
\end{proof}

%\marginpar{
%RM: Can Brendan turn this paragraph into a theorem which asserts nesting
%holds for all $\mu $ and $\nu $ iff the surplus is pseudo-index?}\marginpar{BP: This has now been done in the %math paper.  Do we want to add it here as well/instead, or just refer to it?} 

Note that nestedness is a property of the three-tuple $\left( s,\mu ,\nu
\right) $; that is, for most surplus functions, the model may or may not be
nested depending on the measures under consideration. In some cases,
however, the surplus function is such that the model is nested for all
measures $\left( s,\mu ,\nu \right) $. This is the case for the \emph{%
pseudo-index models} defined above. Indeed, assume that the surplus has the
form:

\begin{equation}
s\left( x_{1},x_{2},y\right) =\alpha \left( x_{1},x_{2}\right) +\sigma
\left( I\left( x_{1},x_{2}\right) ,y\right) .
\end{equation}%
Then 
\begin{equation*}
\frac{\partial s}{\partial y}(x,y)=\frac{\partial \sigma }{\partial y}\left(
I\left( x_{1},x_{2}\right) ,y\right)
\end{equation*}%
only depends on $\left( x_{1},x_{2}\right) $ through the one-dimensional
index $I\left( x\right) $. It follows that the iso-husband sets are defined
by equations of the type $I\left( x\right) =k\left( y\right) $, which do not
depend on $y$. Such curves cannot intersect as soon as $k\left( y\right) $
is strictly monotonic in $y$.\footnote{%
In the un-nested case, it may still be possible to solve the problem using
ad hoc methods. For example, in some cases it can be determined apriori that 
$\gamma $ must couple certain subsets of $X$ with certain subsets of $Y$.
Then, the method above may be applied succesfully to these subsets, even if
it fails when applied to the whole of $X$ and $Y$. We demonstrate this with
an example (see Remark \ref{uniform} below).}$^{,}$\footnote{%
A similar scheme was developed in \cite{P2}, and under strong conditions on
the marginals $\mu $ and $\nu $ and the surplus $s$, it was proven that this
yields the optimal solution.}

In \cite{ChiapporiMcCannPass15m} we show the converse is also true: a
non-degenerate surplus is pseudo-index if and only if $(s,\mu,\nu)$ is
nested for all choices of absolutely continuous population densities $\mu$
and $\nu$.

%Check that $F$ is well defined; that is, that no two level sets chosen above for different values of $y$ intersect.

\subsection{Criteria for nestedness}

The preceding theorem illustrates the powerful implications of nestedness,
when it is present. For suitable data, the next theorem and corollaries give
characterizations of nestedness which are often simpler to check in
practice. These are based on a description of the motion of the iso-husband
set in response to changes in the husband type, which is obtained using the
theory of level set dynamics. Once again, sharper statements and detailed
proofs may be found in \cite{ChiapporiMcCannPass15m}. There the conclusion
of following theorem is also shown to be sharp, in the sense that $dk/dy$
diverges at the endpoints of $Y$ whenever the area of the iso-husband sets $%
F^{-1}(y)$ shrinks to zero. Since we intend to apply the theorem on
Lipschitz domains, we first recall the definition of transversality in that
context.

\begin{definition}[Transversality]
Recall $\overline X(y,k(y))$ intersects $\partial X \in C^1$
non-transversally if their normals coincide (or equivalently, are multiples
of each other) at some point of intersection. If $\partial X$ is merely
Lipschitz, they intersect non-transversally if the inward or outward normal
to $\overline X(y,k(y))$ is a generalized normal to $\partial X$, meaning it
can be expressed as limit of convex combinations of outward normals at
arbitrarily close points where $\partial X$ is differentiable.
\end{definition}

%\marginpar{
%RM: The set $Z$ is defined here in a simpler but more restrictive way than
%in the math paper; it turns out this can cause trouble in some of the
%examples; should we revert to original, more complicated definition?}

\begin{theorem}[Dependence of iso-husbands on husband type]
\label{T:non-nested} Let $X \subset {\mathbf{R}}^m$ and $Y\subset {\mathbf{R}%
}$ be connected Lipschitz domains, equipped with Borel probability measures $%
d\mu(x) =f(x) dx$ and $d\nu(y)=g(y)dy$ whose Lebesgue densities satisfy $%
\log f \in C^{1}(X)$ and $\log g \in C^0_{loc}(Y)$. Assume $s\in C^{3}$ and
non-degenerate throughout $\overline {X} \times Y$. Then the functions $%
k^\pm $ of Theorem~\ref{T:nested} coincide. Moreover $k:=k^\pm \in
C^{1}_{loc}(Y \setminus Z)$ outside the relatively closed set 
\begin{eqnarray}  \label{bad set Z}
Z &:=& \{ y \in Y \mid \overline X(y,k(y)) 
\mbox{ intersects $\p X$
non-transversally} \}.
\end{eqnarray}
As $y \in Y \setminus Z$ increases the outward normal velocity of $%
X_\le(y,k(y))$ at $x \in X(y,k(y))$ is given by $(k^{\prime }- s_{yy})/|D_x
s_y|$. %\footnote{Here $k'=-h_y/h_k$ can be computed from \eqref{grad h}.}
\end{theorem}

\begin{proof}[Sketch of proof]
Under the additional regularity assumed in the theorem, we differentiate %
\eqref{h} at $y \not\in Z$ to obtain 
\begin{eqnarray}  \label{grad h}
h_k := \frac{\partial h}{\partial k} =& \displaystyle \int_{X(y,k)} f(x) 
\frac{d{\mathcal{H}}^{m-1}(x)}{|D_x s_y(x,y)|} > 0 &\mathrm{and} \cr h_y := 
\frac{\partial h}{\partial y} = & \displaystyle- g(y) - \int_{X(y,k)} f(x) 
\frac{s_{yy}(x,y)}{|D_x s_y(x,y)|} d{\mathcal{H}}^{m-1}(x),
\end{eqnarray}
where ${\mathcal{H}}^{m-1}$ is Hausdorff $m-1$ dimensional (surface)
measure. The formula for $h_k$ is a straightforward application of the
co-area formula from \cite{EvansGariepy92}; it can also be seen as a
consequence of the fundamental theorem of calculus, %fact 
after applying the implicit function theorem to $s_y(x,y)=k$ to get $|D_x
s_y|^{-1}$ as the outward normal velocity at $x$ of $X_\le(y,k)$ when $k$ is
increased as $y$ is held fixed. The corresponding formula \eqref{grad h} for 
$h_y$ results from the fact that the analogous Lipschitz velocity is given
by $-s_{yy}/|D_x s_y|$ when $y$ is increased as $k$ is held fixed. Note that
because $y \not\in Z$ the limits from the left and the right which define
the derivatives \eqref{grad h} agree; we need not include any additional
contributions to these integrals coming from $\overline X(y,k) \cap \partial
X$ satisfying $s_y(x,y)=k$. The continous dependence of $h_k$ and $h_y$ on $%
(y,k)$ are shown in \cite{ChiapporiMcCannPass15m}, after which $k\in
C^1_{loc}$ can be inferred by applying the implicit function theorem to $%
h(y,k(y))=0$.
\end{proof}

As consequences of this description, we obtain two alternate
characterizations of nestedness in \cite{ChiapporiMcCannPass15m}:

\begin{corollary}[Dynamic criteria for nestedness]
\label{C:dynamic nestedness} Under the hypotheses of Theorem \ref%
{T:non-nested}: if the model is nested then $k^{\prime }- s_{yy} \ge 0$ for
all $y\in Y\setminus Z$ and $x \in X(y,k(y))$, with strict inequality
holding at some $x$ for each $y$. Conversely, if $Z=\emptyset$ and strict
inequality holds for all $y\in Y$ and $x \in X(y,k(y))$, then the model is
nested.
\end{corollary}

%\begin{proof}
%See \cite{ChiapporiMcCannPass15m}
%\end{proof}

\begin{corollary}[Unique splitting criterion for nestedness]
\label{C:unique splitting} A model $(s,\mu,\nu)$ satisfying the hypotheses
of Theorem \ref{T:non-nested} with $Z=\emptyset$ is nested if and only each $%
x \in X$ corresponds to a unique $y \in Y$ splitting the population
proportionately, i.e. which satisfies 
\begin{equation}  \label{population split}
\int_{X_\le (y, s_y(x,y))} d\mu = \int_{-\infty}^y d\nu.
\end{equation}
In this case, the stable matching is given by ${F}(x)=y$.
\end{corollary}

The first corollary states that the model is nested if and only if all
iso-husband sets move outward as $y$ is increased. Besides proving useful in
the examples below, the second corollary shows nestedness is
methodologically essential: without it, ${F}$ fails to be well-defined
(unless proporitionate population splitting were to be augmented with some
further criteria).

%\marginpar{RM: Could include the proof of Corollary and sharpness remark}

%\begin{proof}
%See \cite{ChiapporiMcCannPass15m}
%\end{proof}

\subsection{Smoothness of payoffs and matchings}

Finally, we are able to address the questions of smoothness  of the
matching function ${F}:X \longrightarrow Y$ in the nested case. In general
dimension, this is a notoriously subtle and challenging question \cite{V2}.
For $n=m>1$ a fairly satisfactory regularity theory has been developed
following work of \cite{mtw}. But for $m>n=1$ little is known, outside of
the pseudo-index case \cite{P2}.

%\marginpar{BP: I might say "little is known" rather than nothing...in the 2012b paper, I prove contininuity under %stong conditions on $(s,\mu,\nu)$.}

%  and the theory for nested models announced here and 
%developed in \cite{ChiapporiMcCannPass15m}.

Assuming transversality ($Z =\emptyset$), Theorems \ref{T:nested} and \ref%
{T:non-nested} give conditions %including $Z=\emptyset$,
under which $F$ is continuous and $k=dv/dk \in C^1_{loc}$ on the interiors
of their respective domains. Recalling $v^{\prime }({F}(x)) = s_y(x,{F}(x))$
from \eqref{foc}, differentiation yields 
\begin{equation}  \label{distributional map gradient}
(k^{\prime }({F}(x))-s_{yy}(x,{F}(x)))D{F}(x) = D_x s_y(x,{F}(x)).
\end{equation}
Thus we immediately see we can bootstrap continuous differentiability of the
matching function $F \in C^1$ from continuity $F\in C^0$, wherever the
normal velocity $k^{\prime }- s_{yy}$ of the iso-husband set of $x$ is \emph{%
strictly} positive. Even assuming $s \in C^\infty$, to get more smoothness
for ${F}$ from this identity, we need more smoothness for $k$, or
equivalently $v$.

Conditions guaranteeing $v \in C^{r,1}_{loc}(Y^0)$ (i.e.\ $r$ times
continuously differentiable, with Lipschitz derivatives) for any integer $%
r\ge 1$ are provided in \cite{ChiapporiMcCannPass15m}. The overall strategy
there is to use an induction on $r$ to extract additional smoothness of the
function $h$ from \eqref{h}, starting with that provided by Theorems \ref%
{T:nested} and \ref{T:non-nested}. Since $h(y,k(y))=0$, this smoothness is
transferred to $k=dv/dy$ using the implicit function theorem. To
differentiate expressions like \eqref{grad h}, we first use a suitably
general form of the divergence theorem to rewrite them as 
\begin{eqnarray*}
h_k &=& \displaystyle \int_{X_\le(y,k)} \nabla_X \cdot V d^m x - \int_{%
\overline{X_\le(y,k)} \cap \partial X} V \cdot \hat n_X d{\mathcal{H}}^{m-1}
\\
h_y &=& \displaystyle- g(y) - \displaystyle \int_{X_\le(y,k)} \nabla_X \cdot
(s_{yy} V) d^m x + \int_{\overline{X_\le(y,k)} \cap \partial X} s_{yy} V
\cdot \hat n_X d{\mathcal{H}}^{m-1}
\end{eqnarray*}
where $V(x)= fs_{yy} \frac{D_xs_y}{|D_xs_y|^2} \Big|_{y={F}(x)}$. Their
derivatives are then given as integrals along the moving interfaces,
weighted by normal velocities as in Theorem \ref{T:non-nested} and its
proof. These in turn must be controlled by appropriate assumptions and a
delicately structured inductive hypothesis.

Finally we note that smoothness often holds for the one-dimensional husbands'  payoff function $v(y)$  but not for the multi-dimensional wives' payoff $u(x)$ or matching function $F(x)$ \cite{ChiapporiMcCannPass15m}.
\subsection{Surplus identification in the nested case}

Assume that we can observe iso-husband sets in a multiple market setting;
what does it tell us about the surplus? We now give a precise answer to that
question. As already noted, if we only observe matching patterns, then the
surplus $s$ can be identified at best up to an additive function of $x$ and
an additive function of $y$. This, however, is not the only flexibility we
have in defining $s$. To see why, remember that an arbitrary iso-husband
set, with an equation of the form $y=F\left( x\right) $, lies in a level set 
\begin{equation*}
\frac{\partial s\left( x,F\left( x\right) \right) }{\partial y}=k
\end{equation*}%
for some constant $k=v^{\prime }(y)$. Knowing the map $F$ for a single pair $%
(\mu,\nu)$ tells the direction (but not the magnitude) of $D_x s_y$ along
the graph of $F$: assuming enough smoothness 
%shown below in Lemma \ref{L:local match}, %and that the inequality ,
it is given by \eqref{distributional map gradient}. To get the magnitude we
need to know the husband's marginal share $v^{\prime \prime }$ of the payoff
also.

Although maps $F$ will not generally exist for all choices of pairs $%
(\mu,\nu)$ (counterexamples are given in \cite{cmn}), if we know the map $F$
corresponding to enough choices of $(\mu,\nu)$, then by a suitable choice we
can make the graph of $F$ pass through any point $(x^{\prime },y^{\prime })$
that we choose. Lemma \ref{L:local match} below shows that this can be done
with sufficient smoothness for \eqref{distributional map gradient} to hold,
and that it costs no generality to take $\mu$ uniform on a small ball around 
$x$. %also shows the inequality \eqref{soc} governing $v''$ 
%can also be made strict
In this way, we learn the direction and sign (but not the magnitude, unless
we also know the marginal payoffs $v^{\prime }$) of $D_x s_y$ globally. 
% in this way.

If we also know the payoffs we can integrate $D_x s_y$ to find $s(x,y)$, up
to an arbitrary additive function $f_0(x) + g_0(y)$ (the constants of
integration). On the other hand, without knowing the payoffs, the direction
of $D_x s_y$ globally is enough to determine the level sets of $s_y(\cdot,
y) $ for each $y$.

Now, the set of continuous functions with the same level sets as a given
function is exactly the set of monotonic transforms of that function. In
other words, a function $G(\cdot )$ has the same level sets as $\frac{%
\partial s}{\partial y}(\cdot ,y)$ if and only if:%
\begin{equation*}
G\left( x\right) =H\left( \frac{\partial s\left( x,y\right) }{\partial y}%
\right)
\end{equation*}%
for some monotonic $H$, possibly depending on $y$. Fixing $y_{0}\in Y$, we
conclude that if $\bar{s}$ is the surplus generating the given iso-husband
sets, then another non-degenerate surplus $s$ generates the same iso-husband
sets if and only if there exists a function $H(z,y)$ with $H_{z}>0$ such
that:%
\begin{equation*}
s\left( x,y\right) =s\left( x,y_{0}\right) +\int_{y_{0}}^{y}H\left( \frac{%
\partial \bar{s}\left( x,t\right) }{\partial y},t\right) dt.
\end{equation*}

This conclusion is implied by the following lemma.

\begin{lemma}[Any couple can be smoothly matched]
\label{L:local match} Fix open sets $X \subset {\mathbf{R}}^m$, $Y \subset {%
\mathbf{R}}$ and $s \in C^{r+1}(X \times Y)$ non-degenerate, with $r \ge 1$.
Given $(x^{\prime },y^{\prime }) \in X \times Y$, and arbitrary small
neighbourhoods $U$ of $x^{\prime }$ and $V$ of $y^{\prime }$, we can find an
atomless probability measure $\nu$ on $V$ such the stable match ${F}\in
C^r(\overline U)$ between $\nu$ and the uniform measure $\mu$ on $U$
satisfies $y^{\prime }={F}(x^{\prime })$, and moreoever: %
\eqref{distributional map gradient} holds for all $x \in \overline U$, and
the husband's payoff $v$ is a quadratic function on $Y$.
\end{lemma}

\begin{proof}
Since the lemma concerns only the behaviour of $s$ near $(x^{\prime
},y^{\prime })$, it costs no generality to assume $X$ and $Y$ bounded and $s
\in C^{r+1}(\overline {X} \times \overline Y)$. Take $v(y)$ to be a
quadratic function of $y$ satisfying 
\begin{equation}  \label{siconcavity}
v^{\prime \prime }(y) = const > \max_{(x,y) \in \overline{X} \times 
\overline{Y}} s_{yy}(x,y)
\end{equation}
and $v^{\prime }(y^{\prime }) = s_y(x^{\prime },y^{\prime })$. Then strict
concavity shows for each $x \in \overline X$, 
\begin{equation}  \label{si transform}
u(x) = \max_{y \in \overline Y} s(x,y) - v(y)
\end{equation}
is uniquely attained by some $y \in \overline Y$, denoted $y={F}(x)$. Strict
concavity also shows that if 
%If $\f(x) \in Y$ then $y=\f(x)$ is also the unique solution
%to
\begin{equation}  \label{sifoc}
s_y(x,y_0) - v^{\prime }(y_0) = 0
\end{equation}
for some $y_0 \in \overline Y$ then $y_0 ={F}(x)$. In particular, $y^{\prime
}={F}(x^{\prime })$, so noting \eqref{siconcavity}, the implicit function
theorem shows $y_0(x)={F}(x)$ is a $C^r$ solution to \eqref{sifoc} on some
neighbourhood $\overline U$ of $x^{\prime }$. Differentiating \eqref{sifoc}
shows \eqref{distributional map gradient} holds on $\overline U$, so
non-degeneracy of $s$ implies $D{F}(x)$ is non-vanishing. 
%and $y_0(x) = \f(x)$ solves \eqref{sifoc}.
% there is a $C^1$ solution $y_0(x)$ to \eqref{sifoc} taking values in $Y$,  with
%Obviously $\f=y_0$ on $\bar U$,
Now we can take $U$ and $V = {F}(U)$ to be as small as we please around $%
x^{\prime }$ and $y^{\prime }$. %with $\p U$ smooth and $\p V \in C^1$.  
We claim ${F}$ is the stable match between the uniform measure $\mu$ on $U$,
and its image $\nu := {F}_\#\mu$, while $u$ and $v$ above are the
corresponding payoffs. The definition \eqref{si transform} shows $(u,v)$ to
be stable \eqref{utilities}, hence integrating $u(x) +v({F}(x)) = s(x,{F}%
(x)) $ against $\mu$ shows $\gamma= (id \times F)_\#\mu$ attains the maximum
and $(u,v)$ attain the minimum in \eqref{duality}. Finally, $\nu$ is
atomless since $DF \ne 0$ shows $F^{-1}(y)$ to be a $C^r$ hypersurface in $%
\bar U$, hence $\mu$ negligible, by the implicit function theorem once more.
\end{proof}

\begin{remark}[Any couple has a strongly nested matching]
If $r \ge 2$ in the lemma above, then taking $U$ smaller if necessary, so $%
\partial U$ intersects each level set of $F$ transversally (or at least ${%
\mathcal{H}}^{m-1}$ negligibly), the co-area formula can be used as in \cite%
{ChiapporiMcCannPass15m} to show $d\nu(y) =g(y) dy$ has a density satisfying 
$\log g \in C_{loc}(V)$. Replacing $(X,Y)$ by $(U,V)$, Corollary \ref%
{C:dynamic nestedness} then shows the model is nested, since $v^{\prime
\prime }(x) - s_{yy}(x,{F}(x))>0$.
\end{remark}

%If the $F$ constructed above fails to be well defined, it is the solution (under some technical conditions). 

%If not, it may be possible to solve the problem using ad hoc methods. 

\section{Applications}

We now consider three applications of our framework.

\subsection{Application 1: Income and fertility}

Our first application considers a model introduced in \cite{L} and \cite{L2}%
. The main question relates to the impact of marriage market considerations
on women's decision to acquire higher education (and more precisely a post
graduate degree). The trade-off Low considers is between human capital and
fertility (what she calls `reproductive capital'). By engaging in post
graduate studies, a woman increases her human capital, which boosts her
future income (among other things). Most of the time, however, a post
graduate degree will require postponing the birth of children to a later
stage of her life, when her fertility may have declined. Husbands, in Low's
model, are interested in both the income of potential spouses and their
fertility; the interaction between these two attributes, and their
consequences on marital patterns are precisely what Low investigates.

In what follows, we generalize Low's model in two directions. First, while
Low assumes that fertility can take only two values (`high' for younger
women and `low' for older, more educated ones), we allow for general
distributions of the two characteristics; in particular, we consider various
correlation patterns between them (from independence to negative
correlation). Second, the preferences we consider are slightly more general
than Low's, in the sense that we allow the presence of children to \emph{%
decrease} utility when the family is poor; in other words, this is a model
in which birth control technologies are either absent or largely imperfect
(several examples can obviously be found in the developing world or in the
history of Western societies). The second feature drastically changes the
qualitative properties of the stable matching; specifically, the model is
nested for some measures but not for others. Moreover, for specific measures
(uniform in our case), both the nested and the unnested cases can be
explicitly solved, which enable us to compare them in a systematic way.

\subsubsection{The model}

\label{incomefertility}Men and women are parameterized by subsets $Y\subset {%
\mathbf{R}}$ and $X\subset {\mathbf{R}}^{2}$, respectively, where

\begin{itemize}
\item $y \in Y$ represents the husband's income,

\item $x_{1}=p$ represents the wife's fertility (probability of having a
child)

\item $x_{2}=x$ is the wife's income.
\end{itemize}

All couples with children receive some lump sum benefit $B$.\footnote{%
A similar example could readily be constructed with either means-tested
benefits or tax rebates instead of a lump sum payment; the only difference
being that the benefit would then be decreasing or increasing with the
couple's income.}

Let $q_{i}$ denote private expenditures of individual $i$ and $Q$ denote
expenditures on children. If the couple has children, their preferences take
the form%
\begin{eqnarray*}
u_{m}\left( q_{m},Q\right) &=&q_{m}\left( Q+1/2\right) \\
u_{f}\left( q_{f},Q\right) &=&q_{f}\left( Q+1/2\right)
\end{eqnarray*}%
whereas without children they are%
\begin{equation*}
u_{g}\left( q_{g}\right) =q_{g},\ \ g=m,f.
\end{equation*}

Suppose man $y$ marries woman $(x,p)$. Then, with probability $p$, they have
a child and maximize their combined utility $\left( q_{m}+q_{f}\right)
\left( Q+1/2\right) $ under the budget constraint: 
\begin{equation*}
q_{m}+q_{f}+Q=x+y+B
\end{equation*}%
If we assume $x+y\geq 1/2-B$, the solution to this maximization problem is 
\begin{equation*}
Q+1/2=q_{m}+q_{f}=\frac{x+y+B+1/2}{2}
\end{equation*}%
generating a total utility equal to $\left( x+y+B+1/2\right) ^{2}/4$. With
probability $1-p$, they do not have a child, and 
\begin{equation*}
u_{m}+u_{f}=q_{m}+q_{f}=x+y.
\end{equation*}

The total utility is therefore equal to the expected value of these two
possible outcomes:%
\begin{equation*}
s\left( p,x,y\right) =p\frac{\left( x+y+B+1/2\right) ^{2}}{4}+\left(
1-p\right) \left( x+y\right) .
\end{equation*}

In what follows, for the sake of simplicity we set the parameter $B$ to $1/2$%
; the solution is therefore defined for all non negative $x$ and $y$, and
the surplus is%
\begin{equation*}
s\left( p,x,y\right) =p\frac{\left( x+y+1\right) ^{2}}{4}+\left( 1-p\right)
\left( x+y\right) .
\end{equation*}

A particular feature of this model is that if parents are poor enough (when $%
x+y<1$, so that $Q<1/2)$ their utility with children is less than without -
implying that no efficient birth control device is available. Such a
situation can simply be ruled out by assuming that $x+y\geq 1$ for all
couples; alternatively, we may consider cases in which some couples are
`poor' ($x+y<1$) whereas others are `wealthy' ($x+y\geq 1$). We consider
both cases in our Examples 1 and 2 respectively; interestingly enough, the
mathematical properties are quite different in the two settings, since one
is nested whereas the other is not.

\subsubsection{Solution}

For this surplus, equation \eqref{envelope} becomes:%
%\[
%\frac{\partial f }{\partial p}=-\frac{1}{p}\left(f \left( p,y\right)
%+y-1\right)
%\]%
%which has general solution:
\begin{equation}
(x+y-1)p=K(y).  \label{gen_inc_fert_sol}
\end{equation}%
where $K(y)=2[\frac{d v(y)}{d y}-1]$. Note, in particular, that if $K\left( 
\bar{y}\right) =0$ for some $\bar{y}$ then all women with $\bar{x}=1-\bar{y}$
marry $\bar{y}$ irrespective of their $p$. In the explicit examples that
follow, we will use the proportionate splitting condition to pin down $%
K\left( y\right) $. For now, we note two general properties: 
%Finally, $K\left( y\right) $ is pinned down by the measure conditions.
%Specifically:

\begin{itemize}
\item the iso-husband set of females $\left( p,x\right) $ for husband $y$ is
given by the potential indifference curve 
\begin{equation*}
x=1-y+\frac{K\left( y\right) }{p}.
\end{equation*}%
In the $\left( p,x\right) $ plane, this curve is decreasing if and only if $%
K\left( y\right) >0$. It intersects lines of constant $p$ and $x$
transversally, as well as the northwest and southeast corners, Since the
other two corners match with the endpoints of $Y$, this 
%so Remark~\ref{R:transversality} 
will imply $Z$ is empty for the rectangular domains considered %below 
in Examples 1 and 2.

%\marginpar{RM: $s$ degenerates at $(p,x,y)=(0,1/2,1/2)$!}

\item In addition, while%
\begin{equation}
\frac{\partial ^{2}s}{\partial x\partial y}=\frac{p}{2}\geq 0
\label{Low sxy}
\end{equation}%
\label{Low spy}we have that%
\begin{equation}
\frac{\partial ^{2}s}{\partial p\partial y}=\frac{x+y-1}{2}.
\end{equation}%
If $x+y\geq 1$, by \eqref{CCond} we expect the level sets of ${F}$ to be
decreasing curves in the $(x,p)$ plane (that is, we expect $K(y)>0$),
meaning husbands face a trade-off of the income versus the fertility of
their spouse. In the opposite case $x+y<1$, these iso-husband curves may be
increasing; we shall see an illustration in our second example.
\end{itemize}

Depending on the measures $\mu $ and $\nu $, the problem may or may not
admit a closed form solution. In what follows, we consider three examples.
First, we provide a complete resolution when the distributions are uniform
on $[0,1]\times \lbrack 1/2,1]$ and $[1/2,1]$, respectively. As a second
example, we consider uniform distributions on $[0,1]\times \lbrack 0,1]$ and 
$[0,1]$.

Lastly, we again take husbands to be uniformly distributed on $[1/2,1]$, but
take the distribution of wives to be uniform on $[1/2,1]\times \lbrack
1/2,3/4]\cup \lbrack 0,1/2]\times \lbrack 3/4,1]$. This distribution is
chosen to reflect an anticorrelation due to age between fertility and
income. Fertility is certainly negatively correlated with age; on the other
hand, women who pursue higher education tend to enter the marriage market
later in life, and so we expect age and education to exhibit a positive
correlation.

\subsubsection{Example 1: uniform and `large' incomes}

%\marginpar{
%RM: may need some revision to eliminate redundancy in light of Theorem \ref%
%{T:nested} and Corollary \ref{C:unique splitting}, depending on the extent
%to which Appendix 1 is revised}

We start with a benchmark case in which the distributions $\mu $ and $\nu $
are uniform on $[0,1]\times \lbrack 1/2,1]$ and $[1/2,1]$ respectively.
Although the surplus degenerates at one corner $(p,x,y)=(0,\frac12,\frac12)$
of $X \times Y$, Proposition \ref{P:optimal match} shows this example is
nested, in view of Corollary \ref{C:unique splitting}. Moreover, the optimal
matching function ${F}(x,p)$ can be solved for almost explicitly; we have
the following result, proved in Appendix \ref{A:Verification of Example 1};
recall that $y$ splits the population proportionately at $(p,x)$ if 
\begin{eqnarray*}
\mu \Big( X_\le (y,s_y(p,x,y))\Big)&=&\mu \Big(\{(\bar{p},\bar{x})\mid \frac{%
\partial s(p,x,y)}{\partial y} \ge \frac{\partial s(\bar{p},\bar{x},y)}{%
\partial y}\}\Big) \\
&=&\nu ([\frac{1}{2},y]).
\end{eqnarray*}%

\begin{proposition}
\label{P:optimal match} Let $\mu$ and $\nu$ be uniform probability measures
on $[0,1] \times [1/2, 1]$ and $[1/2,1]$, respectively. For each $(p,x) \in
[0,1] \times [\frac{1}{2},1]$ there is a unique $y \in [\frac{1}{2},1]$ that
splits the population proportionately at $(p,x)$, and the optimal map takes
the form ${F}(p,x) =y$.
\end{proposition}

Note that, by the proof of this Proposition, we can actually derive an
explicit formula for the optimal map:

%\marginpar{
%RM: Formalize as another Corollary in the preceding section? Or just reorder
%here?} 
\begin{equation*}
{F}(p,x) := \sup_y\Big\{y \mid\mu\Big(X_<(y,s_y(p,x,y)) \Big)>\nu( [\frac{1}{%
2},y))\Big\}.
\end{equation*}

We can go one step further, and identify the iso-husband curve consisting of
the set of points $(p,x)$ which match with a fixed $y$; this is exactly the
potential indifference % iso-husband 
curve \eqref{gen_inc_fert_sol}, for the correct $K(y)$. As is clear from the
proof of Proposition \ref{P:optimal match}, for $y\in \lbrack \frac{1}{2},%
\frac{e}{2(e-1)}]$, we have,

\begin{equation*}
K(y)-\frac{y-1/2}{\ln \left( \frac{y}{y-1/2}\right) }=0.
\end{equation*}

%\marginpar{
%RM: our math paper predict $K'(y) \to \infty$ at the endpoints of $Y$, if
%iso-hushand curves get short}

For $y\in \lbrack \frac{e}{2(e-1)},1]$, we have $K(y)=\bar{x}+y-1$, where $%
\bar{x}$ is the unique solution in the interval $[1/2,1]$ of the equation $%
x-y+(x+y-1)\ln (\frac{y}{x+y-1})=0$. A plot of $K\left( y\right) $ is
provided in Figure 1a, while Figure 1b gives the iso husband curves for
various values of $y$. Notice $K^{\prime }(y)=\bar{x}^{\prime }(y)+1$
diverges like $1/\log y$ at the endpoint $y=1$ where the length of the
iso-husband curve shrinks to zero; the payoff $v(y)$ therefore displays a
comparable singularity in its second derivative. 
%\marginpar{BP: The last sentence in Example 1 is incomplete...}

\ 

\begin{center}
\textbf{Insert Figures 1a and 1b about here}

\ 
\end{center}

\subsubsection{Example 2: uniform, smaller incomes}

\label{uniform}We now consider the same surplus with different measures;
namely, the marginals are uniform on $[0,1]\times \lbrack 0,1]$ and $[0,1]$
respectively. From a mathematical perspective this case provides an
interesting illustration of an un-nested setting. Formally, our method to
derive the optimal match fails in that case; it turns out that there are
more than one $y$ that split the population proportionately for certain
choices $(x,p)$. However, we are able to use a refinement of the argument to
obtain the explicit optimal match. Specifically, we use the symmetry
embedded in the model to show that the optimal map is given by:%
\begin{equation}
G\left( p,x\right) =\left\{ 
\begin{array}{cl}
{F}\left( p,x\right) & \text{ if~}x>1/2, \\ 
1-{F}\left( p,1-x\right) & \text{ if }x<1/2,%
\end{array}%
\right.  \label{symmetricsolution}
\end{equation}%
where ${F}$ is as in the preceding Proposition. This solution displays a
discontinuity along the line $x=\frac{1}{2}$; each wife $(p,\frac{1}{2})$
with income $x=\frac{1}{2}$ is indifferent between two distinct husbands,
one richer ($y=\frac{e^{\frac{1}{p}}}{2(e^{\frac{1}{p}}-1)}$) and the other
poorer ($y=1-\frac{e^{\frac{1}{p}}}{2(e^{\frac{1}{p}}-1)}=\frac{e^{\frac{1}{p%
}}-2}{2(e^{\frac{1}{p}}-1)}$). Although the total surplus generated by the
latter marriage is smaller, her share of it remains the same. The
corresponding isohusband curves are plotted on Figure 2a. We have also
numerically solved this case; the resulting solution agrees with our
theoretical solution and is graphed below (Figure 2b).

\ 

\begin{center}
\textbf{Insert Figures 2a and 2b about here}

\ 
\end{center}

\subsubsection{Example 3: anticorrelated marginals}

In this case, the analytic solution is very complicated, and we provide only
a numerical simulation. For $\epsilon=0$ the numerical solution is plotted
below; one interesting feature is that lower income men (whose iso-husband
curves lie in the green, yellow or orange regions) are indifferent between a
selection of wives in both the low fertility, high income regime, $[0,
\frac12] \times [\frac34,1] $, and the high fertility, low income region $%
[\frac12, 1] \times [\frac12,\frac34+\epsilon]$, whereas the highest income
men, whose iso-husband curves lie in the dark red region, match exclusively
with wives in the high fertility, low income regime. This model does not fit
the hypotheses of Theorem \ref{T:non-nested} when $\epsilon=0$, since $X$ is
neither connected nor Lipschitz, but this shortcoming can be rectified by
taking $\epsilon>0$ arbitarily small.

\ 

\begin{center}
\textbf{Insert Figure 3 about here}

\ 
\end{center}

\subsection{Application 2: A competitive Rochet-Chon\'{e} model}

In our second example, we revisit a seminal model of the literature on
bidimensional adverse selection, due to \cite{rc}. We therefore consider a
hedonic model with:

\begin{itemize}
\item An $n$-dimensional space of products: $z=\left( z_{1},...,z_{n}\right)
\in Z\subset {\mathbf{R}}_{+}^{n}$;

\item An $n$-dimensional space of buyers: $x=\left( x_{1},...,x_{n}\right)
\in X\subset {\mathbf{R}}_{+}^{n}$. Consumers are distributed on this space
according to a probability measure $\mu $. Each consumer will buy exactly
one good. Their utility for purchasing product $z$ for price $P(z)$ is given
by $U\left( x,z\right) -P\left( z\right) $ where 
\begin{equation*}
U\left( x,z\right) =\sum_{i=1}^{n}x_{i}z_{i},
\end{equation*}

\item A one dimensional space of producers: $y\in Y\subset {\mathbf{R}}_{+}$%
. Producers are distributed according to a probability measure $\nu $. Each
producer will produce and sell exactly one good. Their profit in selling a
good $z$ for price $P(z)$ is $P\left( z\right) -c\left( y,z\right) $, where 
\begin{equation*}
c\left( y,z\right) =\frac{1}{2y}\sum_{i=1}^n z_{i}^{2}
\end{equation*}
is the production cost for producer $y$ to produce a good of type $z$.
\end{itemize}

We note that this is exactly the model of \cite{rc}, with the additional
twist that the monopoly producer is replaced with a competitive set of
heterogeneous producers (whose productivity increases in $y$).\footnote{%
Clearly, heterogeneity is not crucial, since the measure on the set $Y$
could be a Dirac. The crucial distinction with Rochet-Chon\'{e} is the
competitive assumption.} As we will see, however, introducing competition
has striking consequences. In particular, and somewhat surprisingly, there
is no bunching in this new variant of their model: consumers of different
types always buy goods of different types. We illustrate these facts with an
example, finding explicitly the optimal matching for certain choices of $\mu 
$ and $\nu $, and then provide a general proof.

\subsubsection{Equilibrium}

As argued above, this hedonic pricing problem is equivalent to the matching
problem on $X\times Y$ associated with the surplus:%
\begin{equation*}
s\left( x,y\right) =\max_{z\in Z}\left( \sum_{i=1}^{n}x_{i}z_{i}-\frac{1}{2y}%
\sum_{i=1}^{n}z_{i}^{2}\right)
\end{equation*}

The maximum in the preceding equation is obtained (assuming $\{yx\mid y \in
Y, x \in X\} \subset Z$) when 
\begin{equation*}
z_{i}=x_{i}y
\end{equation*}%
and so%
\begin{equation*}
s\left( x,y\right) =\frac{1}{2}y\left( \sum_{i=1}^nx_{i}^{2}\right).
\end{equation*}

Now, note that the twist condition is satisfied here; in fact, this surplus
function is of index form, $s(x,y)=S(I(x),y):=\frac{I(x)y}{2}$, where $%
I(x)=|x|^{2}$, and $S$ satisfies the classical Spence-Mirrlees condition $%
\frac{\partial ^{2}S}{\partial I\partial y}>0$. The isoproducer curves take
the form $\sum_{i=1}^{n}x_{i}^{2}=K(y)$, and the matching is monotone
increasing between $|x|$ and $y$.

\subsubsection{An explicit example}

We explicitly solve this problem for a particular choice of measures. Let $%
m=2$ and $\mu$ be uniform measure (normalized to have total mass 1) on the
quarter disk 
\begin{equation*}
\left\{ \left( x_{1},x_{2}\right) \mid x_{1}^{2}+x_{2}^{2}\le 1,x_{1}\geq
0,x_{2}\geq 0\right\}
\end{equation*}
and $\nu$ uniform on $\left[ 1,2\right] $. The stable matching balances the
mass of the quarter disc

\begin{equation*}
\{ x\mid |x|^2 \leq K(y) \}
\end{equation*}
with the $\nu$ mass of $[1,y]$. This implies

\begin{equation*}
\frac{\pi 4K(y)^2 }{4 \pi} =y-1
\end{equation*}
of $K(y) =\sqrt{y-1}$. In other words, the optimal matching takes the form 
\begin{equation*}
{F}(x) = |x|^2 +1.
\end{equation*}
Agent $x$ then buys the product $z$ such that: 
\begin{equation*}
z_{i}=x_{i}\left( \sum_{k=1}^nx_{k}^{2}+1\right) ,\ \ i=1,...,n.
\end{equation*}
The envelope condition implies that agent $x$'s utility satisfies 
%\subsection{Utilities}

\begin{equation*}
\frac{\partial u}{\partial x_{i}}(x)=\frac{\partial s}{\partial x_{i}}\left(
x,{F}\left( x\right) \right)
\end{equation*}%
or, 
\begin{equation*}
\frac{\partial u}{\partial x_{i}}(x)=x_{i}{F}\left( x\right) =x_{i}\left(
1+\sum_{k=1}^{n}x_{k}^{2}\right) .
\end{equation*}%
Similarly, 
\begin{equation*}
v^{\prime }\left( y\right) =\frac{1}{2}\sum_{i=1}^{n}x_{i}^{2}=\frac{y-1}{2}.
\end{equation*}%
Integrating these equations yields 
\begin{eqnarray*}
u\left( x\right) &=&A+\frac{1}{2}\sum_{i=1}^{n}x_{i}^{2}+\frac{1}{4}\left(
\sum_{i=1}^{n}x_{i}^{2}\right) ^{2} \\
v\left( y\right) &=&B+\frac{\left( y-1\right) ^{2}}{4}.
\end{eqnarray*}%
Now, for $y={F}(x)$, we must have $u(x)+v(y)=s(x,y)$. Inserting $y={F}%
(x)=|x|^{2}+1$ into our equation for $v$ yields 
\begin{equation*}
u(x)+v({F}(x))=A+B+\frac{1}{2}|x|^{2}+\frac{1}{4}|x|^{4}+\frac{1}{4}%
|x|^{4}=A+B+\frac{1}{2}|x|^{2}+\frac{1}{2}|x|^{4}.
\end{equation*}%
On the other hand, 
\begin{equation*}
s(x,{F}(x))=\frac{1}{2}(|x|^{2}+1)|x|^{2};
\end{equation*}%
equating these implies $A=-B$. %\begin{eqnarray*}
%u\left( x\right) +v\left( y\right)  &=&A+B+\frac{1}{2}\left( y-1\right) +%
%\frac{1}{2}\left( y-1\right) ^{2} \\
%&=&A+B+\frac{1}{2}y\left( y-1\right) =A+B+s\left( x,y\right) \Rightarrow B=-A
%\end{eqnarray*}
In what follows we assume $A=B=0$; the interpretation is that the firm with
the highest production costs ($y=1$) makes zero profit.

We now turn our attention to the equilibrium pricing schedule, $P(z)$.
Following \cite{cmn}, the stability conditions 
\begin{equation*}
u\left( x\right) +v\left( y\right) \geq s\left( x,y\right) \geq U\left(
x,z\right) -c\left( y,z\right)
\end{equation*}%
imply 
\begin{equation*}
\inf_{y}\left( v\left( y\right) +c\left( y,z\right) \right) \geq P\left(
z\right) \geq \sup_{x}\left( U\left( x,z\right) -u\left( x\right) \right)
\end{equation*}%
or 
\begin{equation*}
\inf_{y}\left( \frac{\left( y-1\right) ^{2}}{4}+\frac{1}{2y}%
\sum_{i=1}^{n}z_{i}^{2}\right) \geq P\left( z\right) \geq \sup_{x}\left(
\sum_{i=1}^{n}x_{i}z_{i}-\frac{1}{2}\sum_{i=1}^{n}x_{i}^{2}-\frac{1}{4}%
\left( \sum_{i=1}^{n}x_{i}^{2}\right) ^{2}\right) .
\end{equation*}%
The solution takes the form 
\begin{equation*}
P\left( z\right) =\bar{P}\left( Z\right) ,\text{ }
\end{equation*}%
where $Z=|z|^{2}$. Note that when agent $y$ sells good $z$, we must have 
\begin{equation*}
P(Z)=\frac{\left( y-1\right) ^{2}}{4}+\frac{1}{2y}\sum_{i=1}^{n}z_{i}^{2}=%
\frac{\left( y-1\right) ^{2}}{4}+\frac{1}{2y}Z.
\end{equation*}%
The first order conditions then imply 
\begin{equation*}
\frac{1}{2}\left( y-1\right) =\frac{1}{2y^{2}}Z,
\end{equation*}%
so that $Z=y^{2}\left( y-1\right) $. We can then solve for $P(Z)$ in terms
of $y$; 
\begin{eqnarray*}
P &=&\frac{\left( y-1\right) ^{2}}{4}+\frac{1}{2y}Z \\
&=&\frac{1}{4}\left( 3y-1\right) \left( y-1\right) .
\end{eqnarray*}%
Noting that $Z=|z|^{2}=y^{2}|x|^{2}=y^{2}(y-1)$ then yields a parametric
representation for the curve $(Z,P(Z))$:

\begin{equation*}
(Z,P(Z))=\left( y^{2}\left( y-1\right) ,\frac{1}{4}\left( 3y-1\right) \left(
y-1\right) \right)
\end{equation*}%
This is plotted on Figure 4.

\begin{center}
\textbf{Insert Figure 4 about here}\ 
\end{center}

\ 

Some comments can be made on this solution. First, and unlike the Rochet-Chon%
\'{e} monopoly case, there is no \textit{exclusion}; all agents buy
products. This property is easily understood from the matching formulation;
as every potential match generates a positive surplus, surplus maximization
implies that no agents will be excluded.\footnote{%
Exclusion is linked to the monopoly logic: by excluding some consumers, the
monopolist can increase the rent received from other buyers.} Note, however,
that the introduction of outside options might reverse this conclusion. More
interesting (and more surprising) is the second finding, namely that \emph{%
there is no bunching of any type}: the efficient allocation is fully
separating so that different agents \emph{always} buy different goods. This
no bunching conclusion is less intuitive; it is not clear, a priori, whether
it holds for general choices of the measures $\mu $ and $\nu $, or it is an
artifact of the particular measures we use in this example. We now present a
theorem which states that the former interpretation is correct.

%\textbf{Question}: is the no bunching result general???

\subsubsection{The general case}

The method for deriving the optimal matching will always work here; as the
surplus has an index form, issues with level sets crossing do not arise. Of
particular interest, the no bunching result is completely general, as the
following result, proved in an appendix, confirms:

\begin{theorem}
\label{nobunching} Different consumer types always buy different goods. That
is, if $x\neq \bar{x}$, and $z$ and $\bar{z}$ represent the respective goods
they choose, then $z\neq \bar{z}$.

\begin{proof}
See Appendix
\end{proof}
\end{theorem}

It can be stressed that there are no restrictions on the measures in this
theorem. In particular, it applies when $\mu $ is uniform measure on the
unit square in ${\mathbf{R}}^{2}$, and $\nu $ is a Dirac mass at $y=1$;
these conditions provide the closest analogue to the Rochet-Chon\'{e}
problem. The difference is that here we have a continuum of homogeneous
producers (each with identical production cost $\frac{|z|^{2}}{2}$ for good $%
z\in {\mathbf{R}}^{2}$) competing with each other, rather than a monopoly.
In other words, the bunching phenomena present in the Rochet-Chon\'{e}
solution are not intrinsically linked to the multidimensionality of the
consumer type space; rather, they are due to market inefficiencies created
by the monopoly.

Finally, it is interesting to note that the framework under consideration
can always be seen as a model of competition under adverse selection; in
that sense, the matching approach provides a natural definition of an
equilibrium in such a framework. A crucial remark, however, is that the
model is characterized by its \emph{private value }nature, since the
producer's profit is not directly related to the identity of the consumer
buying its product (it only depends on the characteristics of the product
and its price). This is in sharp contrast with a common value setting, in
which the buyer's characteristics directly impact the producer's profit.
Think, for instance, of an insurance model a la \cite{RothschildStiglitz},
in which the same insurance contract generates different profit depending on
the buyer's unobserved characteristics (in that case her risk). Technically,
the producer's cost function $c$ now depends on $y$ and $z$ as before, but
also on\emph{\ }$x$. The key point is that, in such a common value context,
the equivalence between matching and hedonic models is lost. In particular,
it is not true in general that there always exists a price function $P\left(
z\right) $ such that a stable matching can be implemented as an hedonic
equilibrium. Our approach, therefore, does construct a first bridge between
the literatures on matching and optimal transportation on the one hand, and 
\emph{competition} under asymmetric information on the other hand; however,
the relationship is specific to private value models. In the non-competitve
setting of the principal-agent framework, the analogous connection found in 
\cite{Carlier01} and \cite{FigalliKimMcCann11} has proven exteremely
fruitful.

\subsection{Application 3: a simple hedonic model}

As a last illustration, we consider a simple hedonic model, in which
consumers have heterogeneous tastes for differentiated products, which are
produced by firms with heterogeneous cost functions. The basic framework is
similar to the previous, Rochet-Chon\'{e} one\footnote{%
In particular, we keep utilities of the form $\sum_{i}x_{i}z_{i}$, where $%
x=\left( x_{1},...,x_{n}\right) $ denotres an agent's idiosyncratic
valuations of product characteristics. For empirical applications, $x$ will
typically be considered as a random vector.}, except for one feature: the
cost function is now:%
\begin{equation*}
c\left( y,z\right) =\sum_{i=1}^{n}\frac{z_{i}^{2}}{2y_{i}}
\end{equation*}%
where the vector $y=\left( y_{1},...,y_{n}\right) $ is producer-specific.
The multidimensionality of $y$ reflects the fact that different producers
made have different comparative advantages in producing \emph{some} of the
characteristics but not others; for instance, a producer may be good at
producing fast cars but less efficient for smaller ones.\footnote{%
We maintain here the assumption that each firm produces one good. This can
readily be relaxed provided that production functions are linear in \emph{%
quantity} produced (while quadratic in characteristics $z$).} Note that,
unlike most of the empirical IO literature (starting with the seminal
contributions by \cite{BerryLevinsohnPakes95} and \cite%
{BerryLevinsohnPakes04} ), which assume imperfect competition, our hedonic
framework posits that producers are price takers. In that sense, our
matching framework can be seen as an alternative approach to modeling
competition in differentiated products.

The characterization of the equilibrium follows the same path as before. The
surplus function is:%
\begin{equation*}
s\left( x,y\right) =\max_{z\in Z}\sum_{i=1}^{n}\left( x_{i}z_{i}-\frac{%
z_{i}^{2}}{2y_{i}}\right).
\end{equation*}%
Here, the maximum is obtained when $z_{i}=x_{i}y_{i}$, leading to:%
\begin{equation*}
s\left( x,y\right) =\frac{1}{2}\sum_{i=1}^{n}x_{i}^{2}y_{i}
\end{equation*}

Since $s$ is continuous, existence follows from Theorem \ref{Xist}.
Moreover, this form is a particular case of the general structure studied in
subsection 2.4. It follows that the twist condition is satisfied whenever $%
x\in {\mathbf{R}}_{+}^{n}$, an assumption we maintain throughout the current
subsection. This guarantees uniqueness and purity:\footnote{%
In fact, the restriction of $x$ to the positive orthant is not essential,
since the twist condition holds for all $x$ in the complement of the
coordinate hyperplanes $x_i=0$. The latter form a set of measure zero which it is always
possible to omit, as in \cite{cmn}.} there exists a function ${F}$ such that 
$x$ is matched with $y={F}\left( x\right) $. The simplicity of these
existence and uniqueness arguments is indeed an important advantage of the
matching approach.

As mentioned above, purity implies that different agents are matched with
different producers. Moreover, since $z_{i}=x_{i}y_{i}$, a result by
Lindenlaub implies that different agents always buy different products; in
other words, there is no bunching in this model. Clearly, these conclusions
follows from the assumption that the dimension of heterogeneity is the same
for individuals and firms. If this assumption is relaxed (for instance by
assuming that $n<m$), then a continuum of different buyers purchase from the
same producer, and the market shares can be analyzed using the results in
Section 3.

The specific form of function ${F} $ depends on the distributions of
individual and firm characteristics. As previously, we illustrate our
approach by explicitly solving this problem for a particular choice of
measures. Let $\mu $ be the uniform measure (normalized to have total mass
1) on the disk $\sum_{i=1}^{n}\left( x_{i}-a_{i}\right) ^{2}\leq 1$ and let $%
\nu $ be the uniform measure (again normalized to have total mass 1) on the
disk $\sum_{i=1}^{n}\left( y_{i}-b_{i}\right) ^{2}\leq 1$ with $%
a_{i},b_{i}>1,i=1,...,n$. Then the support of the stable measure is born by
the function ${F} =\left( {F} _{1},...,{F} _{n}\right) $:%
\begin{equation*}
y_{i}={F} _{i}\left( x\right) =x_{i}-a_{i}+b_{i}
\end{equation*}%
Agent $x$ then buys the product $z$ such that: 
\begin{equation*}
z_{i}=x_{i}\left( x_{i}-a_{i}+b_{i}\right)
\end{equation*}%
and firm $y$ produces the product $z$ such that: 
\begin{equation*}
z_{i}=\left( a_{i}-b_{i}+y_{i}\right) y_{i}
\end{equation*}

Utilities are derived as before:

\begin{eqnarray*}
\frac{\partial u}{\partial x_{i}}(x) &=&\frac{\partial s}{\partial x_{i}}%
\left( x,{F} \left( x\right) \right) \\
&=&x_{i}\left( x_{i}-a_{i}+b_{i}\right)
\end{eqnarray*}%
therefore 
\begin{equation*}
u\left( x\right) =K+\sum_{i=1}^{n}\left( \frac{x_{i}^{3}}{3}-\frac{%
a_{i}-b_{i}}{2}x_{i}^{2}\right)
\end{equation*}%
and similarly%
\begin{equation*}
v\left( y\right) =K^{\prime }+\frac{1}{6}\sum_{i=1}^{n}\left(
a_{i}-b_{i}+y_{i}\right) ^{3}
\end{equation*}%
with $K+K^{\prime }=0.$

The hedonic price satisfies: 
\begin{equation*}
\inf_{y}\left( v\left( y\right) +c\left( y,z\right) \right) \geq P\left(
z\right) \geq \sup_{x}\left( U\left( x,z\right) -u\left( x\right) \right)
\end{equation*}%
which gives: 
\begin{equation*}
P\left( z\right) =-K+\frac{1}{12}\sum_{i=1}^{n}\left[ \left( \left(
a_{i}-b_{i}\right) ^{2}+4z_{i}\right) ^{\frac{3}{2}}+\left(
a_{i}-b_{i}\right) \left( \left( a_{i}-b_{i}\right) ^{2}+6z_{i}\right) %
\right].
\end{equation*}

As usual in a quasi-linear framework (i.e. in the absence of income
effects), the price schedule is defined up to an additive constant $K$; the
latter can be pinned down, for instance, by imposing conditions on the
profit of the least productive firm.

\section{Conclusion}

This paper provides a general characterization of multidimensional matching
models, in terms of existence, uniqueness and qualitative properties of
stable matches. We show that recent developments in the literature on
multidimensional optimal transport can be exploited to derive conditions
under which stable matches are unique and pure, and to understand their
local geometry. We consider a first application of such models to a
competitive version of a standard, hedonic model in which heterogeneous
consumers purchase differentiated commodities from heterogeneous producers.

Of specific interest are situations in which the dimensions of heterogeneity
on the two sides of the market are unequal. We explore the topology of the
`indifference sets' that arise in this setting, and provide conditions under
which they can be expected to be smooth manifolds of dimension $m-n$. In
particular, we investigate the set of `multi-to-one dimensional matching
problems', and we introduce a nestedness criterion under which the
equilibrium match can be found more or less explicitly. Lastly, we stress
the deep relationships that exist between the multi-to-one dimensional
framework and models of competition under multidimensional asymmetric
information, at least in the case of private values. Considering a
competitive variant of the seminal, Rochet-Chon\'{e} (1998) problem, in
which goods can be produced by a $1$-dimensional, heterogeneous distribution
of producers, rather than a single monopolist, we provide a full
characterization of the resulting equilibrium price schedule. In particular,
we show that in our competitive framework, and in contrast to the original
monopolist setting, there is never bunching; that is, consumers of different
types \textit{always} buy goods of different types.

%\bibliographystyle{plain}%{mf}
%\bibliography{biblio}
%\end{document}

\begin{appendices}

%\section{Proof that twist implies purity}
%\label{A:twist implies purity}

\section{Verification of Examples 1 and 2}
%Example \#1 is nested but \#2 isn't}
\label{A:Verification of Example 1}
%Proof of optimality of the construction}

We begin with Proposition \ref{P:optimal match}. 

\begin{proof}
We verify that for each $(p,x) \in (0,1) \times (1/2,1)$, there is a unique $y \in (1/2,1)$ splitting the population proportionally at $(p,x)$; the result then follows from Corollary \ref{C:unique splitting}.

The proportionally splitting level curves $X(y,k(y))$ take the form $x^y(p)=1-y+\frac{K(y)}{p}$, where $K(y)$ is chosen to satisfy the population splitting condition.   Consider now the way that the curves $X(y,k(y))$  intersect the boundary; as each $p \mapsto x^y(p)$ is continuous on $(0,1)$ and monotone decreasing (note that $K(y) =(x+y-1)p \geq 0$), it  intersects either  $[0,1] \times \{1/2\}$ in a unique point $(p,x) = (p(y),1/2)$ or   $\{1\} \times [1/2,1]$ in a unique point $(p,x) = (1,x(y))$.
We will prove the following two properties:

\begin{enumerate}
\item $K(y)$ is monotone increasing in $y$.
\item The boundary intersection points have a certain monotonicity property.  Precisely, for each $y_0<y_1$, one of the following occurs:

\begin{enumerate}
\item $X(y_0,k(y_0))$ intersects $[0,1] \times \{1/2\}$ and $X(y_1,k(y_1))$ intersects $\{1\} \times [1/2,1]$. 
\item $X(y_0,k(y_0))$ and $X(y_1,k(y_1))$ both intersect $[0,1] \times \{1/2\}$, and $p(y_0)<p(y_1)$.
\item $X(y_0,k(y_0))$ and $X(y_1,k(y_1))$ both intersect $\{1\} \times [1/2,1]$, and $x(y_0) <x(y_1)$.
\end{enumerate}
\end{enumerate}

These two facts will imply the desired result as follows: 1)  will imply that $x^{y_1}(p) -x^{y_0}(p)$ is decreasing in $p$, for fixed $y_0 <y_1$, as the derivative of this function is $x^{y_1}(p) -x^{y_0}(p) =\frac{K(y_0)-K(y_1)}{p^2} <0$.  This means that if two of the population splitting curves intersect, (that is $x^{y_1}(\bar p) -x^{y_0}(\bar p)=0$) within the given domain, then $x^{y_1}(p) <x^{y_0}(p)$ for all $p > \bar p$.  This in turn implies that
the boundary intersection points satisfy one of the following:

\begin{enumerate}[a)]
\item $X(y_1,k(y_1))$ intersects $[0,1] \times \{1/2\}$ and $X(y_0,k(y_0))$ intersects $\{1\} \times [1/2,1]$. 
\item $X(y_0,k(y_0))$ and $X(y_1,k(y_1))$ both intersect $[0,1] \times \{1/2\}$, and $p(y_0)>p(y_1)$.
\item $X(y_0,k(y_0))$ and $X(y_1,k(y_1))$ both intersect $\{1\} \times [1/2,1]$, and $x(y_0) >x(y_1)$.
\end{enumerate}
This clearly contradicts point 2) above and so establishes the desired result.

To complete the proof, then, it remains only to verify points 1) and 2).  We first consider population splitting curves which pass through $[0,1] \times \{1/2\}$.  In this case, the proportional splitting condition is given by

\begin{eqnarray*}
y - 0.5 &=& \int_{0.5}^1 \frac{p(y)(\bar x +y -1)}{(x+y-1)}dx 
\\ &=& (\bar x +y -1)p(y)[\ln(1+y-1)-\ln(0.5+y-1)] ]
\\ &=& (y -0.5) p(y) \ln(\frac{y}{y-0.5}). 
\end{eqnarray*}
This means that $p(y) = \frac{1}{ \ln(\frac{y}{y-0.5})}$.  A simple calculation shows that this function is increasing in $y$.  The function can also be inverted to obtain $y(p) =\frac{e^{\frac{1}{ p}}}{2(e^{\frac{1}{ p}}-1)}$; this tells us that the population splitting level curves pass through $[0,1] \times \{1/2\}$ precisely for $y\leq y(1) = \frac{e}{2(e-1)}$.  For $y$ in this region, the monotonicity of $p(y)$ implies that 2) b) holds, and we note that $K(y) =\frac{y-1/2}{\ln(\frac{y}{y-1/2})}$, which is also montone increasing, verifying 1) in this region.

Therefore, for $y \geq \frac{e}{2(e-1)}$, the curve $X(y,k(y))$ intersects $\{1\} \times [1/2,1]$.  This confirms part 2) a); to complete the proof,  it remains only to show that $x(y)$ and $K(y)$ are monotone increasing on $[\frac{e}{2(e-1)},1]$.  In this region, the proportional splitting condition gives us:

\begin{eqnarray*}
1-y &=& \int_{ x(y)}^1 1-\frac{1(x(y) +y -1)}{(x+y-1)}dx 
\\ &=& 1- x (y)-(x(y) +y -1)[\ln(1+y-1)-\ln( x(y)+y-1) 
\end{eqnarray*}
or, equivalently, $x(y)$ is the unique solution in $[0.5,1]$ of 
\begin{equation}
0 =f(x,y):= x-y +(x +y -1)\ln(\frac{y}{x+y-1}).
\end{equation}
  Differentiating implicitly, we have
\begin{equation}\label{implicit_differentiation}
x'(y) =-\frac{f_y}{f_x}.
\end{equation}

We have, for $x<1$
$$
f_x = 1+\ln(\frac{y}{ x +y-1})-1=\ln(\frac{y}{ x +y-1})>0
$$
(as, for $x<1$, $\frac{y}{ x +y-1}>1$).

We now show $f_y <0$.   We have

\begin{equation}
f_y=-2+\ln(\frac{y}{ x +y -1})+\frac{ x +y -1}{y}.\label{yderivative}
\end{equation}

Now, as $ x  \leq 1$, the last term  satisfies
\begin{equation}\label{lastterm}
\frac{ x +y -1}{y} \leq 1
\end{equation}
 (with equality only when $ x =1$), and the second to last term is decreasing in $y$, so it is less than its value at  $y=\frac{e}{2(e-1)}$:
\begin{eqnarray}
\ln(\frac{y}{ x +y -1})& \leq&  \ln(\frac{e}{2(e-1)})-\ln( \frac{e}{2(e-1)}+ x-1)\\
&\leq &\ln(\frac{e}{2(e-1)})-\ln( \frac{e}{2(e-1)}-\frac{1}{2})\\
& = & %\ln(\frac{e}{2(e-1)})-\ln(\frac{1}{2(e-1)})=
1,\label{secondterm}
\end{eqnarray}
where the second inequality follows, as  $ x \geq \frac{1}{2}$.  Note that we have equality above only if $y = \frac{e}{2(e-1)}$ and $ x=\frac{1}{2}$.

Now note that one of the equalities \eqref{lastterm} or \eqref{secondterm} is always strict (as either $ x <1$ or $x > \frac{1}{2}$).  Therefore, the derivative \eqref{yderivative} is negative on the relevant range.

As $f_y <0$ and $f_x>0$, \eqref{implicit_differentiation} tells us $x'(y)>0$.  Finally, we have $K(y) =x(y)+y-1$, and so the monotonicity of $x(y)$ tells us $K(y)$ is also strictly increasing on this region.
\end{proof}

Finally, we turn our attention to uniform measures on $[0,1]$ and  $[0,1]^2$, and the proof of formula \eqref{symmetricsolution} in Section \ref{uniform}.  The proof requires the following Lemma.
\begin{lemma}\label{splitting}
Assume $\mu$ is uniform on  $[0,1]^2$ and $\nu$ is uniform on $[0,1]$. If $(p,x,y)$ is in the support of the optimal matching $\gamma$, and $p>0$ then either $x \geq \frac{1}{2}$ and $y \geq \frac{1}{2}$ or $x \leq \frac{1}{2}$ and $y \leq \frac{1}{2}$.
\end{lemma}
\begin{proof}
Note that $s(p,1-x,1-y) -s(p,x,y)=2-x-y$ exhibits no interaction terms.  As the measures are symmetric under the transfomation $(p,x) \rightarrow(p, 1-x)$ and $y \rightarrow 1-y$, the (unique) stable matching will be symmetric under the transformation $(p,x,y) \rightarrow (p,1-x,1-y)$.  

Now, let $(p,x,y)$ belong to the support of $\gamma$; then, by invariance, we must also have $(p,1-x,1-y)$ in the support of $\gamma$.  Applying $2$ $s$ monotonicity to these points yields

\begin{equation*}
s(p,x,y) + s(p,1-x,1-y) \geq s(p,1-x,y) + s(p,x,1-y).
\end{equation*}
Now note that for fixed $p$, the function $(x,y) \mapsto s(x,y,p)$ satisfies the Spence-Mirrlees condition, $\frac{\partial ^2s}{\partial x \partial y} =\frac{p}{2} >0$.  This implies 

\begin{equation*}
\big(x-(1-x)\big)\big(y-(1-y)\big) =(2x-1)(2y-1)\geq 0
\end{equation*}
which is equivalent to the stated result.
\end{proof}

We now prove formula \eqref{symmetricsolution} in Section \ref{uniform}.

\begin{proof}
By the preceding Lemma, the optimal map maps the region $[0,1] \times [\frac{1}{2},1]$ to $[\frac{1}{2},1]$ and  $[0,1] \times [0,\frac{1}{2}]$ to $[0,\frac{1}{2}]$.  The map in the first region then must be given by $G(p,x) =f(p,x)$, and the mapping in the second region is $G(p,x) =1-f(p,1-x)$ by symmetry.
%To construct it, we need to find the optimal map between $\mu$, uniform measure on $[0,1] \times [\frac{1}{2},1]$, and $\nu$, uniform measure on $[\frac{1}{2},1]$; by the invariance property $F(p,x) =1-F(p,1-x)$, we will then %get the mapping in the other region for free.
\end{proof}

\section{Proof of no bunching}
\label{A:Proof of no bunching}

Here we prove the no bunching result, Theorem \ref{nobunching}, for our competitive variant of
Rochet and Chon\'{e}'s screening model.

\begin{proof}
We prove that $x \neq \bar x$ always buy different goods.  Let $y$ and $\bar y$ be the buyers they match with, respectively, and assume that $x$ and $\bar x$ both buy the good $z$ (from $y$ and $\bar y$ respectively).   We will show this implies $x =\bar x$.

We have that $z =y x =\bar y \bar x$.

As $(x,y)$ and $(\bar x,\bar y)$ belong to the support of the optimal measure $\gamma$, they satisfy the $2$-monotonicity condition:

$$
s(x,y) +s(\bar x, \bar y) \geq s(x,\bar y) +s(\bar x,  y) 
$$

Now, we have that  $s(x,y)=x \cdot z -\frac{1}{2y}|z|^2$, $s(\bar x,\bar y)=\bar x \cdot  z -\frac{1}{2\bar y}| z|^2$ and $s(x,\bar y) \geq x \cdot z -\frac{1}{2\bar y}|z|^2$, $s(\bar x, y) \ge \bar x \cdot  z -\frac{1}{2 y}| z|^2$.  Substituting into the above equation yields:

$$
x \cdot z -\frac{1}{2y}|z|^2+ \bar x \cdot  z -\frac{1}{2\bar y}| z|^2\geq s(x,\bar y) +s(\bar x,  y) \geq  x \cdot z -\frac{1}{2\bar y}|z|^2+\bar x \cdot  z -\frac{1}{2 y}| z|^2
$$

The right and left hand side are identical, so we must have equality throughout.  In particular, we have that
$$ 
s(\bar x, y) = \bar x \cdot z -\frac{1}{2 y}|z|^2.
$$
This implies $z=y \bar x$ (as $z$ maximizes the joint surplus for $y$ and $\bar x$).  But we also have $z =y x$ from above.  Therefore

$$
y \bar x = y x.
$$

Canceling the $y$ gives the desired result.
\end{proof}
\section{Figures}

    \includepdf[pages={1,2,3,4,5,6}]{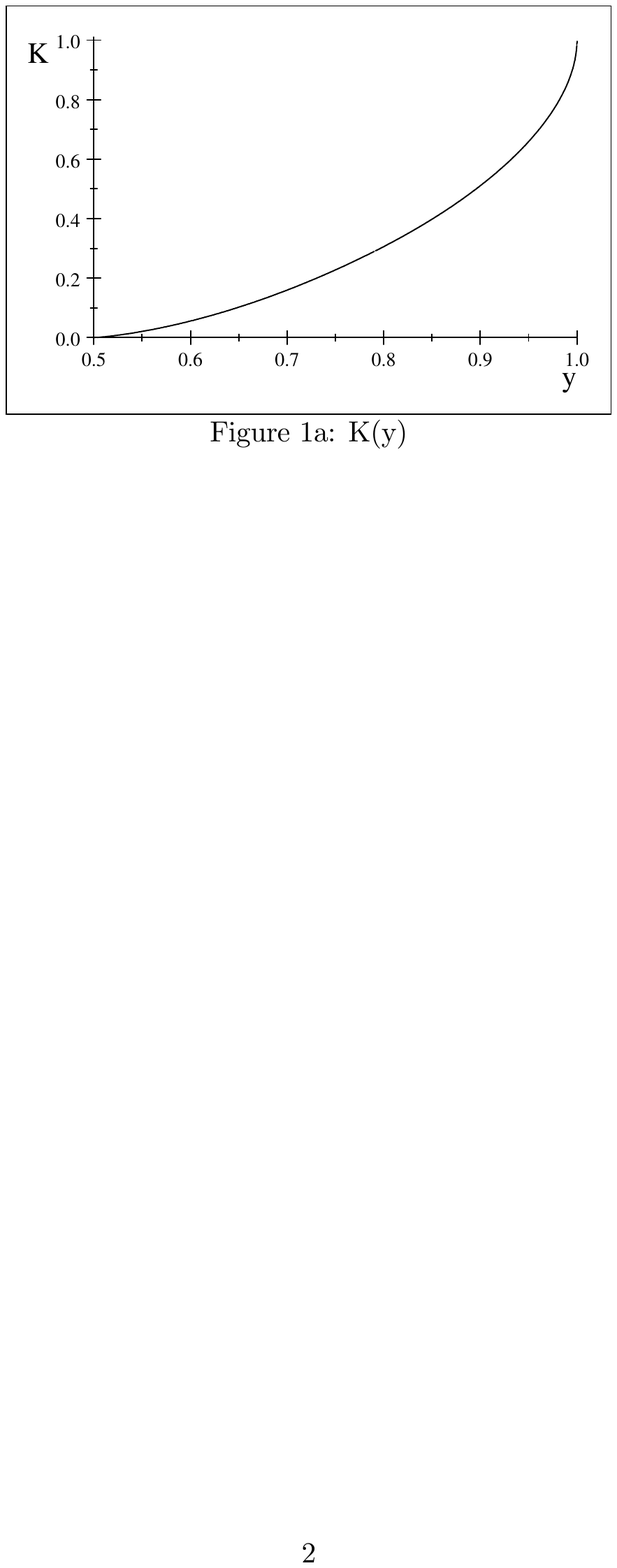}

\end{appendices}

\end{document}